\documentclass[aps,prb,twocolumn,10pt,superscriptaddress]{revtex4-2}
\usepackage{header}
\usepackage{adjustbox}
\usepackage{physics}
\newtheorem{theorem}{Theorem}
\newtheorem{lemma}[theorem]{Lemma}

\newtheorem{definition}[theorem]{Definition}
\newtheorem{proposition}[theorem]{Proposition}

\begin{document}
\title{Foliated Quantum Error Correction for Qudits}
\author{G\"ozde \"Ust\"un}
\affiliation{Centre for Quantum Software and Information, University of Technology Sydney, Sydney, New South Wales 2007, Australia}
\affiliation{Joint Center for Quantum Information and Computer Science, University of Maryland, College Park, MD 20742, USA}
\author{Simon J. Devitt}
\affiliation{Centre for Quantum Software and Information, University of Technology Sydney, Sydney, New South Wales 2007, Australia}
\affiliation{InstituteQ, Aalto University, 02150 Espoo, Finland}
\author{Jason Saied}
\affiliation{Sydney, Australia}

\begin{abstract}
We present a framework for foliating any Pauli-based quantum error-correcting code over prime-dimensional qudits. For any such code, we obtain a qudit graph state that can be measured to perform fault-tolerant measurement-based quantum computing. 
Such a paradigm is of interest in platforms such as photonics, where measurement-based protocols are natural and high-dimensional states are readily available. 
We discuss several examples for arbitrary prime dimension $d$, such as the qudit toric code (stabilizer, CSS), the $d$-dimensional perfect $[[5,1,3]]$ code (stabilizer, non-CSS), and a straightforward $d$-dimensional generalization of the CSS honeycomb code (dynamical, CSS). 
Under a simple error model, we numerically calculate thresholds for the foliated qudit toric code and demonstrate that they are comparable to the non-foliated version. 
\end{abstract}
\maketitle
\maketitle

Measurement-based quantum computation (MBQC)~\cite{MBQC, raussendorf2003measurement}
is a paradigm for universal quantum computation in which a large graph state is prepared, then processed using single-qubit (or qudit) measurements. 
MBQC and related paradigms such as fusion-based quantum computing~\cite{Bartolucci2023, bombin2021interleaving, bombin2023faulttolerantcomplexes} 
are of particular interest in platforms such as photonics: since photons cannot easily interact, it is beneficial to separate the (Clifford) entanglement generation step from the (non-Clifford) measurements. The choice of measurements largely determines the computation. 
These same features also make MBQC appealing for protocols such as blind and verifiable quantum computing~\cite{Broadbent_2009, morimae2012blind, hayashi2015verifiable, fitzsimons2017unconditionally}. 
Fault-tolerant versions of MBQC are well-known in the qubit case, beginning with Raussendorf's 3D cluster state providing a \emph{foliation} of the surface code \cite{Raussendorf2005, RAUSSENDORF20062242, Raussendorf2007, Raussendorf_topological}. 
That is, the two-dimensional code is replaced by a series of connected two-dimensional sheets, with the third dimension representing the flow of information over time. 
The code has detectors, related to the surface code stabilizers, that can be used for quantum error correction, with comparable performance to the standard surface code. 
This construction was generalized to provide foliations of general (stabilizer) CSS codes in \cite{bolt2016foliated}. The work of \cite{brown2020universal} then generalized this to the case of non-CSS and subsystem codes (in the qubit setting). 

We now consider the qudit case. 
In many instances, the error-correction threshold of a quantum error-correcting code has been seen to \emph{increase} with the dimension of the underlying quantum system \cite{andriyanova2012new, anwar2014fast, watson2015fast, Brock_2025}. 
However, this theoretical increase is not universally advantageous: moving beyond qubits typically leads to more complex noise models with higher error rates and greater control challenges~\cite{Ringbauer2022, Low2025}. 
In photonic systems, however, higher-dimensional states are readily available, and the dominant noise mechanism remains photon loss, regardless of the system dimension. 
This motivates interest in fault-tolerant MBQC over higher-dimensional systems, where one may potentially take advantage of the greater thresholds of qudit codes. 
MBQC has been developed for qudit systems in the non-fault-tolerant case~\cite{Zhou_2003, booth2023outcome, Romanova_2026}. 

In the present work, we give a construction for foliation of an arbitrary Pauli-based quantum error-correcting code over prime-dimensional qudits. 
This generalizes the construction of \cite{brown2020universal} to qudits, and to the case of \emph{dynamical} codes rather than only stabilizer or subsystem codes.
After establishing notation and conventions, we review the foundational case of foliation of a single physical qudit. 
We then present Definition~\ref{def:css graph state}, which gives the general construction of the foliated graph state corresponding to a (possibly dynamical) CSS code. The remainder of the section discusses the relevant detectors that can be used for quantum error correction. We also briefly discuss the small modifications that must be made for the non-CSS case. 
We then give examples, starting from the foliated qudit toric code, for which we numerically demonstrate comparable error rates to the work of \cite{watson2015fast} for circuit-based codes, with the error-correction threshold increasing with the qudit dimension. 
As an example of a non-CSS code, we discuss the foliated $[[ 5,1,3]]$ qudit code \cite{chau1997five}. Finally, as a dynamical example, we give a qudit generalization of the CSS honeycomb code \cite{davydova2023floquet} and demonstrate its foliated realization. We support and verify our theoretical constructions with simulations.

\textit{Qudit Paulis and Qudit Graph States}
We briefly recall and fix notation for the essential properties of the qudit Pauli operators and graph states in prime dimension $d$. 
(See \cite{helwig2013absolutelymaximallyentangledqudit} for a more detailed exposition using similar notation.)  
Let $\omega$ be a primitive $d$th root of unity, and write the computational basis states of the qudit as $\ket{0}, \dots, \ket{d-1}\in\mathbb{C}^d$. 
We have
    $Z = \sum_{k=0}^{d-1} \omega^k \ketbra{k}, \,\, X = \sum_{k=0}^{d-1} \ketbra{d+1}{d},$  
satisfying $XZ=\omega^{-1} ZX$. 
We also recall the Fourier transform $F = \frac{1}{\sqrt{d}}\left(\omega^{ij}\right)_{0\leq i,j\leq d-1}$, 
which satisfies
$FZF^\dagger = X^\dagger,\,\,
FXF^\dagger = Z$. 
The operators $X$ and $Z$ generate the (qudit) Pauli group (often called the Heisenberg-Weyl group), with elements $X^a Z^b$ for $a,b\in\mathbb{Z}_d$. (We will generally neglect global phase factors.) 
Since the dimension $d$ is prime, the standard notions and results regarding stabilizer groups generalize to the qudit setting, replacing arithmetic modulo $2$ with arithmetic modulo $d$ \cite{gottesman1998fault}. 
Correspondingly, we will represent a multi-qudit Pauli operator $X^{p_{X,1}}Z^{p_{Z,2}}\otimes \cdots\otimes X^{p_{X,r}}Z^{p_{Z,r}}$ by a pair of vectors over $\mathbb{Z}_d$, $p_X = (p_{X,1}, \dots, p_{X,r})$ and $p_Z=(p_{Z,1}, \dots, p_{Z,r})$

We recall the CZ gate, which has the following equivalent characterizations: 
\begin{equation*}
    CZ = \sum_{a,b=0}^{d-1} \omega^{ab}\ketbra{ab} = \sum_{a=0}^{d-1} \ketbra{a}\otimes Z^a = \sum_{b=0}^{d-1} Z^b \otimes \ketbra{b}.
\end{equation*}
We will use the notation $CZ_{i,j}$ to refer to the CZ gate acting on qudits $i$ and $j$ of a larger state. Since the CZ gate is symmetric, this notation is unambiguous. 

Let $G$ be a simple graph with nodes $\{1, \dots, N\}$ and edge set $E$, and further let each edge $(i,j)\in E$ be assigned a weight $w_{ij}\in \{1, \dots, d-1\}$. 
(If $(i,j)\not\in E$, we say that $w_{ij}=0$.) 
The corresponding qudit graph state is the state
\begin{equation*}
    \prod_{(i,j)\in E} (CZ_{i,j})^{w_{ij}}\ket{+}^{\otimes N},
\end{equation*}
where $\ket{+}=\frac{1}{\sqrt{d}}\sum_{k=0}^{d-1}\ket{k}$. 
This may be characterized as a stabilizer state, with stabilizer group generated by the following $N$ stabilizers, one for each node $i$: 
\begin{equation} \label{eq:main}
    X_i \prod_{j} Z_j^{w_{ij}}.
\end{equation}
This gives a simple recipe for measuring a stabilizer $\prod_j Z_j^{p_j}$ on a set of qudits: we prepare an ancilla in the $\ket{+}$ state, apply $CZ^{p_j}$ between qudit $j$ and the ancilla, then measure $X$ on the ancilla. 
This is equivalent to the standard circuit-based measurement circuit and will be the key to measuring the checks of our foliated code. 

\textit{Qudit Error Correction}
We will consider several levels of generality for qudit error-correcting codes. 
The most commonly studied are \emph{stabilizer codes}, in which the generators of an abelian group $\mathcal{S}$ of Pauli operators are measured and used to detect and correct errors. 
We also have the more general family of \emph{subsystem codes}, in which we measure the generators of a (not necessarily abelian) \emph{gauge group} $\mathcal{G}$ of Pauli operators. 
The stabilizers used for error detection/correction are then the center of the gauge group; the values of the stabilizers are inferred by appropriately combining gauge generator measurement outcomes \cite{poulin2005stabilizer}. 
Finally, we will also consider \emph{dynamical codes}, a generalization of subsystem codes in which the choice of checks can change between rounds, potentially in an aperiodic fashion. In the present work, we discuss only unmasked stabilizers \cite{fu2025error}, but the construction may be straightforwardly modified to account for different choices of measurement schedule. 

We now formalize our notation for CSS-type codes (possibly dynamical) involving $N$ data qudits of dimension $d$. 
(We focus on this case for simplicity, and briefly discuss the modifications necessary for the non-CSS case later.) 
For $0\leq t\leq T$, we assume we are given a $r(t)\times N$ generator matrix $H(t)$ (with $r(t)$ being the number of checks measured in time step $t$).  
For $t$ even, the rows of $H(t)$ specify the $Z$-type checks measured in step $t$, with the $(i,j)$ entry indicating the power of $Z_j$ in the $i$th check. 
For odd $t$, the rows of $H(t)$ analogously give the $X$-type checks. For a subsystem (or stabilizer) code, the generator matrices $H(t)$ only depend on the parity of $t$, with $H(2k)=H_Z$ and $H(2k+1)=H_X$, where $H_X$, $H_Z$ are the usual parity check matrices. 
We also assume each data qudit's initial state is specified as some eigenstate of $X$ or $Z$. (This assumption is not essential and can be removed, as discussed in the following section.) 

\textit{Single Qudit Foliation:}
We begin by discussing a trivial but crucial special case, the foliation of a single \emph{physical qudit}. Teleportation generalizes straightforwardly to the (prime-dimensional) qudit case, as in Fig.~\ref{fig:teleportation}. 
Thus one may encode a state $\ket{\psi}$ in a measurement-based manner by preparing a physical qudit in state $\ket{\psi}$, along with a linear chain of $\ket{+}$ states, and appropriately performing $CZ^{\pm 1}$ gates and $X$ measurements. 
We refer to a state of the resulting form (with an even number of physical qubits) as a \emph{chain state} $Ch(\ket{\psi})$, as in Figure~\ref{fig:chain}; by abuse of notation, we generally avoid specifying the length of the chain. 

\begin{figure}
    \centering
    \includegraphics[width=\linewidth]{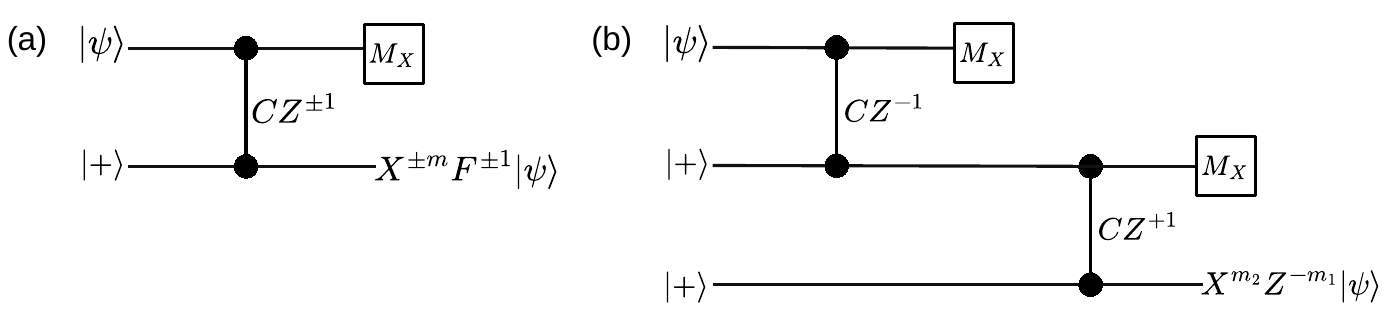}
    \caption{(a) Teleportation of a single qudit state using a $CZ^{\pm 1}$ gate and $X$ basis measurement. The output state is Fourier transformed and has a Pauli factor depending on the outcome $m$ of the measurement. 
    (b) Composing two teleportation circuits with opposite powers of $CZ$, we see that the Fourier factors cancel. Here $m_1, m_2$ are the measurement outcomes. 
    }
    \label{fig:teleportation}
\end{figure}
\begin{figure}
    \centering
    \includegraphics[width=\linewidth]{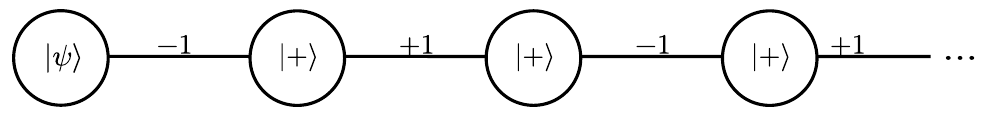}
    \caption{A linear chain state $Ch(\psi)$, encoding the state $\ket{\psi}$ in a measurement-based manner.}
    \label{fig:chain}
\end{figure}

In what follows, however, it is cleaner to always prepare the same linear chain $Ch(\ket{+})$ beginning with only $\ket{+}$ states, then initialize other chains $Ch(\ket{\psi})$ using measurements. 
As noted, the default linear chain already has the form $Ch(\ket{+})$; 
further, one may modify this to any state $Ch(Z^k\ket{+})$ by simply applying $Z^k$ to the first qudit, or tracking this correction in a \emph{Pauli frame} analogous to the qubit case. 
It is also simple to prepare a chain $Ch(\ket{k})$ by starting with a longer chain $Ch(\ket{+})$, then measuring the first two qudits in the $Z$ and $X$ bases respectively. (Up to a known Pauli frame, due to the randomness of the measurements.) 
The state is progressed along the chain by performing $X$ measurements on every data qudit except for the final two. At the end of a chain, one measures the $X$ observable by measuring the final two qudits in the $X$ and $Z$ bases respectively; the $Z$ observable is measured using the bases $Z$ and $X$ respectively. 

In the constructions below, each \emph{data qudit} of the code will be foliated into a chain state. 
Within each layer (or ``time step''), we further entangle the data qudits with ancillas. 
The ancillas are measured to obtain the values of the code checks in a given step, as discussed below \eqref{eq:main}. The post-measurement state of the data qudits is then teleported to future layers in the chain state. 

\textit{Framework for Qudit Measurement-Based Quantum Computing (CSS):}
We now give the precise description of the graph state and detectors needed for foliated qudit MBQC in the CSS case. 
For convenience, we will index the nodes of the graph state by $3$-tuples, where the first index indicates whether it is a data or ancilla qudit, the second indicates its label, and the third indicates the time step. 
\begin{definition}\label{def:css graph state}
Given a CSS-type code specified by $d, N, T$, and $H(t)$ as above, we construct the corresponding graph state $\Lambda=\Lambda(d,N,T,H)$: 
\begin{enumerate}
    \item We have the set of \emph{data qudit nodes} 
        $\mathcal{D} = \{(0,q,t): 1\leq q \leq N, 0\leq t \leq T\}$.
    For $t>0$, we add an edge from $(0,q,t-1)$ to $(0,q,t)$, with weight $(-1)^t$. 
    \item For $0\leq t \leq T$, we have the set of \emph{ancilla qudit nodes}
        $\mathcal{A}_t = \{(1, c, t): 1\leq c \leq r(t)\}$,
    with each $c$ corresponding to a check at time $t$ (a row of $H(t)$). 
    For each such node $(1,c,t)$, and each data qudit $q$ in the support of $c$, 
    we add an edge to $(0,q,t)$ with weight $(-1)^t H(t)_{cq}$. 
\end{enumerate}
\end{definition}
To perform (a memory experiment in) MBQC, one simply measures the $X$ observable on every qudit, modulo the initial and final measurements as discussed above. 
The result of measuring $X$ on ancilla $(1,c,t)$ roughly corresponds to measuring the check $c$ at time step $t$ in the analogous circuit-based code. 
(See the discussion below \eqref{eq:main}.) 
However, due to the measurement-based nature, the exact values of the measurements will be random even without errors. 
Thus (analogous to the case of circuit-based QEC with measurement errors), one must form detectors using (roughly) the product of consecutive measurements of the same stabilizer. 
Rigorously, a \emph{detector} is a stabilizer of the graph state $\Lambda$ that commutes with all the measurements performed, so that (regardless of the random measurement outcomes) it stabilizes the post-measurement state. 
Cases in which the detectors \emph{do not} stabilize the final state imply the presence of errors, allowing for error detection and correction. 
We begin with the stabilizer case, since it is the simplest (compare with Fig.~\ref{fig:examples} for the case of the toric code): 
\begin{definition}\label{def:detector}
For each ancilla node $(1,c,t)$ with $t+2\leq T$, the associated detector is
\begin{equation}\label{eq:detector explicit}
    D(c,t)=X_{(1,c,t)}X_{(1,c,t+2)}^{-1} \times \left(\prod_{q=0}^{N-1} X_{(0,q,t+1)}^{H(t)_{cq}}\right)^{(-1)^t}.
\end{equation}
\end{definition}
This detector has a simple form, forced by the structure of the qudit graph state: we have the initial ancilla with weight $+1$, and the ``same'' ancilla in a later time step (the next time it is measured!) with weight $-1$. (This is what one would do in a circuit-based framework as well, so that the product would always have eigenvalue +1 in the absence of errors.) In our MBQC case, we must also involve the data qudits in between, with powers corresponding to the power with which they appear in the relevant check. Depending on the initialization and terminal measurement conditions, one may have additional time boundary detectors, truncated versions of the above, which we leave as a straightforward exercise for the reader. 

More generally, consider a stabilizer $g$ expressed as a product of checks: $g = c_1^{p_1}\cdots c_{r}^{p_{r}}$, where the $c_j$ are the checks measured at time $t$ (viewed as Pauli group elements) and $r=r(t)$. 
(As this is the non-dynamical case, the time $t$ only indicates whether $g$ is composed of Pauli $X$ or $Z$ operators.) 
There is a detector corresponding to $g$ at time $t$, and it is simply the appropriate product of the $D(c_j,t)$: 
\begin{equation}\label{eq:detector product}
        D(g,t):=D(c_1,t)^{p_1}\cdots D(c_{r},t)^{p_{r}}. 
\end{equation}
One typically does not use such detectors for error correction with a stabilizer code, as they are redundant with the $D(c_j,t)$. However, this perspective is useful to lead into the subsystem case. 

In a subsystem code, recall that the checks $c_j$ (rows of $H(t)$) are generators of the gauge group, and the stabilizers are the center of the gauge group. 
For any such stabilizer $g = c_1^{p_1}\cdots c_{r}^{p_{r}}$ expressed as above, the detector $D(g,t)$ corresponding to $g$ and starting at time $t$ is given by \eqref{eq:detector product}. 
In particular, it is irrelevant whether the individual $c_j$ are stabilizers themselves (equivalently, whether the $D(c_j,t)$ are detectors); since the appropriate product of the $c_j$ is a stabilizer, the corresponding product of the $D(c_j,t)$ is a detector. 
 
For the dynamical case, let $S(t)$ be the instantaneous stabilizer group at time $t$ (that is, the stabilizers of the code space \emph{after} measuring the checks $c_1, \dots, c_r$ of $H(t)$). 
We say a stabilizer is \emph{measured at time $t$} if it can be expressed in terms of the checks $c_j$ of $H(t)$, so that its value can be inferred from the check measurements. 
We may build a detector corresponding to any element in $\bigcap_{j=0}^k S(t+j)$ that is measured at times $t$ and $t+k$. 
These detectors are similar to the above, but they must be ``stretched'' in the time direction so that they begin and end in time steps $t, t+k$ which measure the stabilizer. 
(Note, in a dynamical code, not every element of $S(t)$ is actually measured at time $t$, and may perhaps never be measured at all. 
The stabilizers that are measured are referred to as ``unmasked'' stabilizers \cite{fu2025error}.) 

We briefly consider the effect of Pauli errors after the preparation of the graph state. 
Since our detectors consist only of $X$ measurements, only $Z$ errors are relevant to the error correction. In fact, we have an equivalence with the corresponding circuit-based code under phenomenological noise, with circuit-level $Z$ and $X$ errors mapping to $Z$ errors on the graph state in even and odd time steps respectively. We make this correspondence precise in the Appendix. 

\textit{Non-CSS case:} 
There are the following three changes from the CSS case: (1) Rather than ancillary qudits living in a single time step, they will in general be adjacent to \emph{both} $Z$-type qudits in step $t=2k$ \emph{and} $X$-type qudits in step $t=2k+1$.  
(2) For \emph{every} pair of ancillary qudits measured during the same pair of time steps $(2k,2k+1)$, 
corresponding to checks $X^{p_X}Z^{p_Z}$ and $X^{q_X}Z^{q_Z}$, 
we need an additional edge of weight $p_X\cdot q_Z= p_Z\cdot q_X$. This serves to preserve the commutativity between the checks measured in the same round (cf. \cite{brown2020universal}). 
(3) On the ancillary qudit corresponding to a check $X^{p_X}Z^{p_Z}$, one must measure $XZ^{-p_X\cdot p_Z}$ instead of just $X$.  
This directly generalizes the CSS case, in which each check has $p_X=0$ or $p_Z=0$. 
The non-CSS case is discussed more formally in the Supplementary Materials, Sec.~\ref{sec:non css theory}. 

\textit{Examples: } Here, for the purposes of illustration, we briefly discuss foliated qudit analogues of the toric code, the $5$-qudit perfect code, and the CSS honeycomb code. Further details can be found in the Appendix and Supplementary Material.

\begin{figure}
    \centering
    \includegraphics[width=\linewidth]{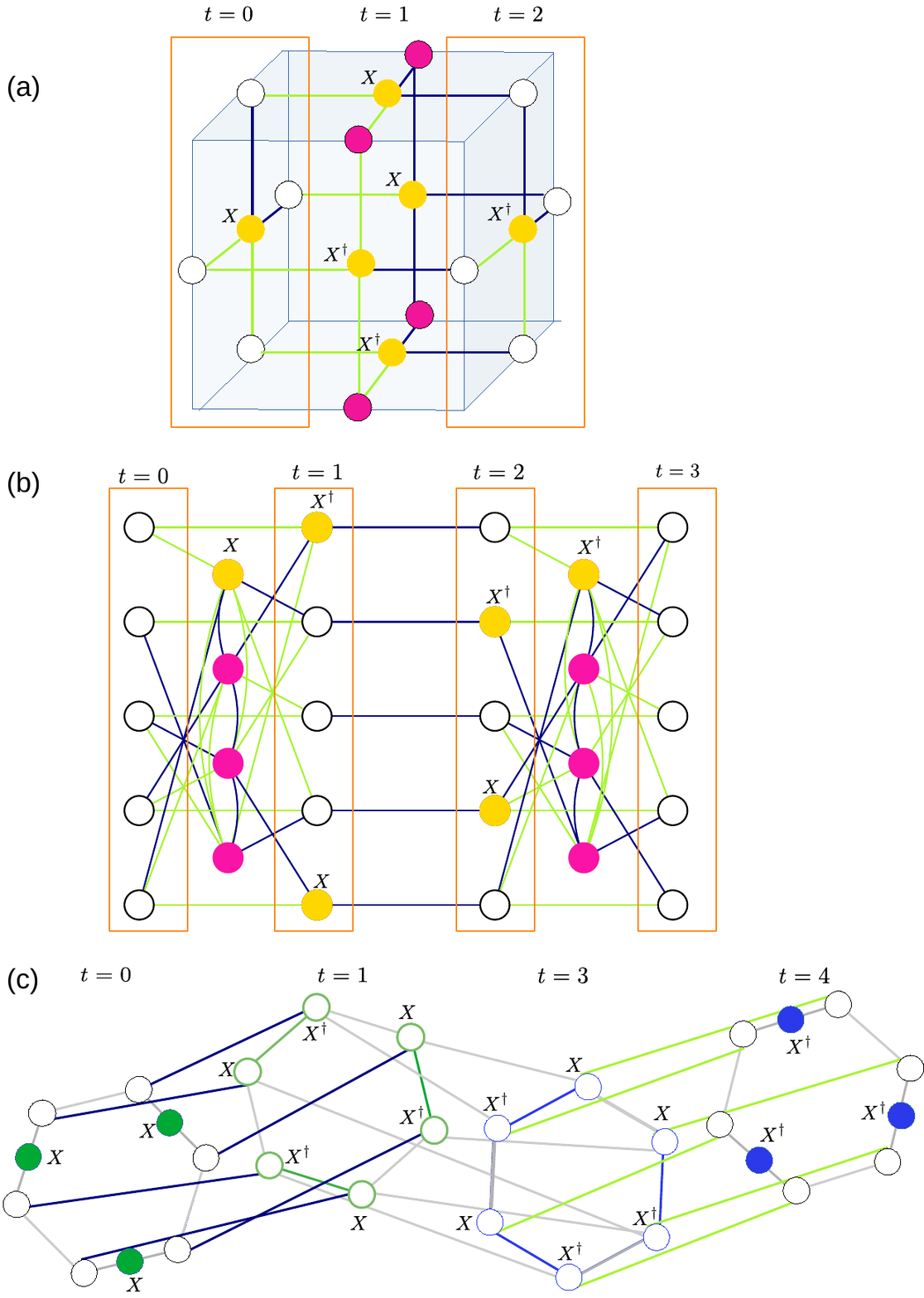}
    \caption{Examples of foliated qudit codes. 
(a) shows the unit cell for the qudit toric code, 
(b) shows two measurement rounds (four time steps) of the foliated (qudit) perfect code, and 
(c) shows part of the foliated CSS honeycomb code. 
In all examples, darker lines correspond to edges with weight $+1$, while green lines correspond to edges with weight $-1$. 
In (a) and (b), the golden nodes (labeled with Paulis) depict an example detector of the foliated code. (With the uninvolved ancillas colored pink.) 
The colors in (c) are used differently to avoid conflicting with the honeycomb code color conventions. 
We show enough of the graph state to exhibit the detector corresponding to a red face of the honeycomb code. 
Green and blue are used to mark the ancilla nodes corresponding to the green and blue edges of that face. All but six of the pictured nodes are involved in the detector. (Note that no nodes at $t=2$ are involved, so none are shown.) 
}
    \label{fig:examples}
\end{figure}

In Fig.~\ref{fig:examples}(a), we present the \emph{unit cell} of the qudit toric code \cite{Bullock_2007}. Similar to the qubit case, the graph state for the foliated qudit toric code may be obtained by tiling the cubic unit cell in three dimensions (with toric boundary conditions in two directions). The Pauli operators in the figure are used to depict the $Z$-type detector $D(c,0)$ hosted on this unit cell. (The $X$-type detectors are hosted \emph{between} unit cells.) 
For the purposes of validation, in Appendix~\ref{app:toric numerics} we calculate thresholds for the foliated qudit toric code under a simple noise model in which Pauli errors are applied after preparation of the graph state. 
For $d=3$, we find threshold approximately $0.023$; this value increases with $d$, with $d=7919$ giving threshold approximately $0.088$. (See Table~\ref{tab:example}.) 
These thresholds are in fact larger than the analogous values in \cite{watson2015fast}, as we discuss in the Appendix. 

In Fig.~\ref{fig:examples}b, we present an example of a foliated non-CSS code, specifically the qudit $[[5,1,3]]$ code \cite{chau1997five}, a perfect code stabilized by $Z^\dagger \otimes X^\dagger\otimes I\otimes X\otimes Z$ and its cyclic permutations. 
The figure shows the full foliation over two layers of the code. Unlike the CSS case, the ancillas are entangled with data qudits in both even and odd time steps, and there are edges between the ancillas.

In Fig.~\ref{fig:examples}c, we exhibit part of the foliation of the qudit CSS honeycomb code as an example of a dynamical qudit code. 
As the qudit analogue of the CSS honeycomb code \cite{davydova2023floquet} was not present in the literature, we discuss this generalization in further detail in Appendix~\ref{app:honeycomb}. 
Due to the dynamical measurement schedule, more time steps are required to complete a detector compared to the previous two examples. 
The figure illustrates the detector associated with a red face of the honeycomb code. 
A red face involves three green edges, three blue edges, and six data qudits. 
Thus, the detector involves three green ancilla qudits at $t = 0$, the relevant six data qudits at $t=1$ and $t=3$, and finally the three blue ancilla qudits at $t=4$. 
We note that this detector is not supported on the $t=2$ portion of the graph state. 

\textit{Conclusion}
In this work, we present a general framework for constructing high-dimensional measurement-based quantum computing (MBQC) that applies to any family of Pauli-based QEC codes—including stabilizer, subsystem, and dynamical codes—in both the CSS and non-CSS cases. As an example of a CSS stabilizer code, we foliate the qudit toric code and numerically demonstrate that its error-correction threshold is comparable to the circuit-based version, using the decoder of \cite{anwar2014fast, watson2015fast}. 
We then present the foliated perfect $[[5,1,3]]$ qudit code as an example of a non-CSS family. We also present a new generalization of the CSS honeycomb code \cite{davydova2023floquet} to the qudit case, and discuss its foliation as an example of a CSS Floquet code. 

Fault-tolerant qudit MBQC is of particular interest in platforms such as photonics, where measurement-based protocols are natural and high-dimensional states are readily available. 
For realistic implementation of qudit MBQC, we note that it is unnecessary to construct the entire graph state ahead of time. Instead, one only needs to keep $2-3$ layers of the state in memory at a time, enough to measure the relevant checks and teleport to the next layer. Subsequent layers may be entangled in later, as they are needed. 
In the photonic setting, where the deterministic $CZ$ gates required to entangle with subsequent layers are inaccessible, one can instead break the single large graph state into smaller uniform pieces, which can be generated in a deterministic~\cite{gimeno2019deterministic, raissi2024deterministic,comparing_schemes} or multiplexed \cite{lee2023graph,pankovich2024flexible,lobl2025transforming} fashion. 
These smaller pieces are then joined and measured using destructive entangling operations such as \emph{fusion} \cite{browne_rudolph, ustun2025fusion, comparing_schemes}. 
This is the essential idea behind fusion-based quantum computation (FBQC)~\cite{Bartolucci2023, bombin2021interleaving, bombin2023faulttolerantcomplexes}, which presently has been developed only for qubits. 
Future work should focus on modifying the foliated MBQC construction to allow for resource-efficient FBQC in the qudit case. 

\section*{DATA AVAILABILITY STATEMENT}
\label{sec_data}
The simulations were performed using \texttt{sdim}~\cite{kabir2026sdimquditstabilizersimulator}, and the corresponding data are available in the auxiliary files of the arXiv version of the paper. 
\section*{Acknowledgments}
Project led by University of Technology Sydney and supported by Defence Science and Technologies Group (DSTG) and Advanced Strategic Capabilities Accelerator (ASCA) through its Emerging and Disruptive Technologies (EDT) Program.
This work was supported in part by ARO grants W911NF-23-1-0242 and W911NF-23-1- 0258.
\newpage
\section*{Appendix}
\setcounter{section}{0}
\renewcommand{\thesection}{\Alph{section}}
\renewcommand{\theequation}{\arabic{equation}}
\renewcommand{\thefigure}{\arabic{figure}}
\renewcommand{\thetheorem}{\arabic{theorem}}

\section{Noise Model and Equivalence to Circuit-Based Phenomenological Noise}
\begin{proposition}
    Given a CSS-type code specified by $d, N, T$, and $H(t)$ in Definition~\ref{def:css graph state}, the MBQC protocol on the graph state $\Lambda=\Lambda(d,N,T,H)$ is logically equivalent to a circuit-based implementation of the code. Further, Pauli $Z$ errors on ancilla qudits of $\Lambda$ correspond to measurement errors in the circuit-based framework. Pauli $Z$ errors on data qudits $(0,q,t)$ of $\Lambda$ correspond to data qudit $Z$ or $X$ errors if $t$ is even or odd respectively. 
\end{proposition}
\begin{proof}
    In MBQC, we measure checks on data qudits by applying powers of $CZ$ gates between them and an ancilla in the $\ket{+}$ state, then measuring the ancilla in the $X$ basis. 
    Writing $\ket{+} = F\ket{0}$ and $M_X = F^\dagger M_Z$, then rewriting $(1\otimes F^\dagger)CZ(1\otimes F)=CX$ (controlled $X$ gates targeting the ancilla), we obtain the standard circuit for extracting $Z$ syndromes in the circuit-based framework. 
    In odd time steps, the data qudits have been conjugated by Fourier transforms due to the teleportations before and after the syndrome extraction. 
    Since we conjugate a $Z$-check measurement circuit by Fourier transforms, we obtain an $X$-check measurement circuit. 
    From this description, we directly see that a Pauli $Z$ error on an ancilla qudit affects the corresponding check and corresponds to a measurement error. 
    Similarly, a Pauli $Z$ error on a data qudit in time step $t$ will flip the measurement outcomes of all checks measured in the \emph{subsequent} time step $t+1$. 
    Translating to the circuit-based model, this is equivalent to a $Z$ error if $t$ is even (since the subsequent layer measures $X$-type stabilizers) and an $X$ error if $t$ is odd. 
\end{proof}
In our simulations, we consider a simple noise model in which the graph state $\Lambda$ is prepared perfectly, then Pauli $Z$ errors are applied to each qudit with probability $p$. More precisely, for $1\leq k\leq d-1$, we apply $Z^{k}$ with probability $p/(d-1)$. 
Due to the Proposition, we note that Pauli $X$ errors (after graph state preparation) are irrelevant, as they do not affect any of the detectors in the CSS case. 
This is not unique to the qudit setting. 
\section{Qudit Toric Code Numerics}\label{app:toric numerics}
We consider the (rotated) qudit toric code of even distance $D$ (with $d$-dimensional qudits). The qudits sit on a $D\times D$ square lattice with toric boundary conditions. The stabilizers live on the plaquettes of the lattice and are essentially the same as the qubit case, with $X$ replaced by $X^\dagger$ on the left side of each plaquette and $Z$ replaced by $Z^\dagger$ on the bottom of each plaquette. 

For decoding the qudit toric code, we implemented a variant of the Hard-Decisions Renormalization Group (HDRG) decoder of \cite{anwar2014fast, watson2015fast}. 
This is a CSS decoding algorithm taking place on the ($X$ or $Z$ type) \emph{syndrome graph}, in which nodes are detectors and edges correspond to single-qudit Pauli errors that flip the relevant pair of detectors. 
The qudit toric code's decoding graph has a simple lattice structure: we use this structure to create an easily-calculated distance metric between nodes in the syndrome graph. 
We give an informal summary of the operation of the HDRG decoder, given a \emph{syndrome} consisting of a number in $\mathbb{Z}_d$ for each detector: 
(1) For any $2$ adjacent detectors with opposite syndromes, we infer that an error has occurred on the edge between them; record this correction and update the syndrome. 
(2) Identify clusters of ``nearby'' detectors with nonzero syndrome. (Each node in the cluster should be at most distance $r$ from some other node in the cluster.) 
(3) For each ``neutral'' cluster, in which the syndromes add up to $0$, there exists a Pauli correction that eliminates the cluster (turning all of its syndromes to $0$). Identify one such correction, record it, and update the syndrome. 
(4) If there are still nonzero syndromes, return to step $2$, this time increasing the value of $r$. 

The sketch above (since it includes step 1) corresponds to the HDRG decoder with one level of \emph{initialization}, in which ``obvious'' errors are corrected before the clustering begins. 
We note that there is a small ambiguity in Step 3: especially as the permitted distance $r$ between nodes in a cluster grows larger, there are often multiple inequivalent choices of Pauli correction to eliminate a cluster. 
We did not find a discussion on how this choice was made in \cite{anwar2014fast, watson2015fast}, so our implementation may have improved upon this choice slightly. For each neutral cluster of detectors with nontrivial syndromes (adding to $0$), we constructed a complete graph whose nodes corresponded to the detectors in the cluster and whose edges were weighted with the distance between the nodes. We then removed edges to obtain a minimum-weight spanning tree for the cluster and only chose corrections corresponding to those edges. This leads to lower-weight corrections which may be logically inequivalent to (say) a random choice. 

As validation of foliated qudit MBQC, we simulated the foliated qudit toric codes under Pauli $Z$ noise using \texttt{sdim}~\cite{kabir2026sdimquditstabilizersimulator}, then used HDRG decoding to calculate the logical error rates and thresholds. We give some calculated thresholds for different qudit dimensions in Table~\ref{tab:example}, and we plot the logical error rate curves leading to the $d=7919$ threshold in Fig.~\ref{fig:decoder}. 
The work \cite{anwar2014fast} applied HDRG decoding to (unrotated) qudit toric codes with only data qudit errors, and \cite{watson2015fast} to (unrotated) qudit surface codes with both data and measurement errors. 
Thus there is no perfect circuit-based comparison for our results, but the closest comparison is to the initialization level 1 curve in Figure 6 of \cite{watson2015fast}, which uses an equivalent error model to ours but considers the unrotated surface code rather than the rotated toric code. 
We find significantly \emph{higher} thresholds than those found in previous work. 
We emphasize that this is \emph{not} due to the difference between MBQC and circuit-based constructions, as the syndrome graphs used for decoding are identical. 
It is likely partially due to the difference between the rotated toric code and unrotated surface code; they should theoretically have the same asymptotic thresholds, but small-size effects in simulation can often lead to different calculated threshold values. 
However, for the ``limiting'' case $d=7919$, our calculated threshold is nearly double that of \cite{watson2015fast}. Thus we believe the main difference is due to our modified version of Step 3 above: by choosing to correct clusters according to a minimum-distance spanning tree, we may be finding lower-weight corrections than in previous work. 

\begin{figure}[H]
    \centering
    \includegraphics[width=.8\linewidth]{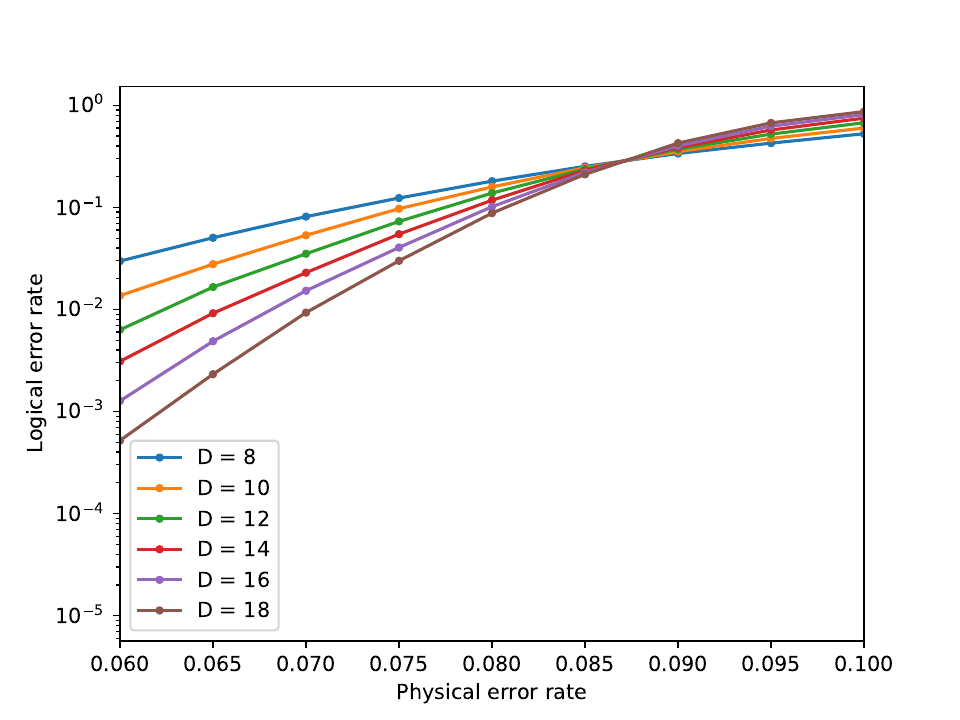}
    \caption{The logical vs. physical error rate of the foliated toric code of qudit dimension $d=7919$, at different lattice sizes $D$. }
    \label{fig:decoder}
\end{figure}

We also note that readers unfamilar with qudits may be surprised that the threshold increases with the qudit dimension. 
Recall that in e.g. minimum-weight perfect matching decoding of \emph{qubit} toric codes, one must often deal with \emph{strings} of (say) Pauli $Z$ errors, in which several adjacent errors create cancelling syndromes that are only nontrivial at the beginning and end of the string. In the qudit case, there are $d-1$ \emph{different} nonzero powers of $Z$ that may occur, so the intermediate syndromes in an error string will rarely cancel. This makes even fairly long strings of errors much easier to correct in the qudit case. 



\begin{table}
\centering
\caption{Simulated error threshold for the qudit toric code of dimension $d$ using our implementation of the HDRG decoder. }
\label{tab:example}
\begin{tabular}{|c|c|} 
 \hline
 \textbf{Qudit dimension d} & \textbf{Error Threshold} \\
 \hline
 3    & 0.023 \\
 \hline
 5    & 0.032  \\
 \hline
 7    & 0.041  \\
 \hline
 7919 & 0.088 \\
 \hline
\end{tabular}
\end{table}

\section{CSS Honeycomb Code for Qudits}\label{app:honeycomb}

\begin{figure} [H]
    \centering
    \includegraphics[width=.8\linewidth]{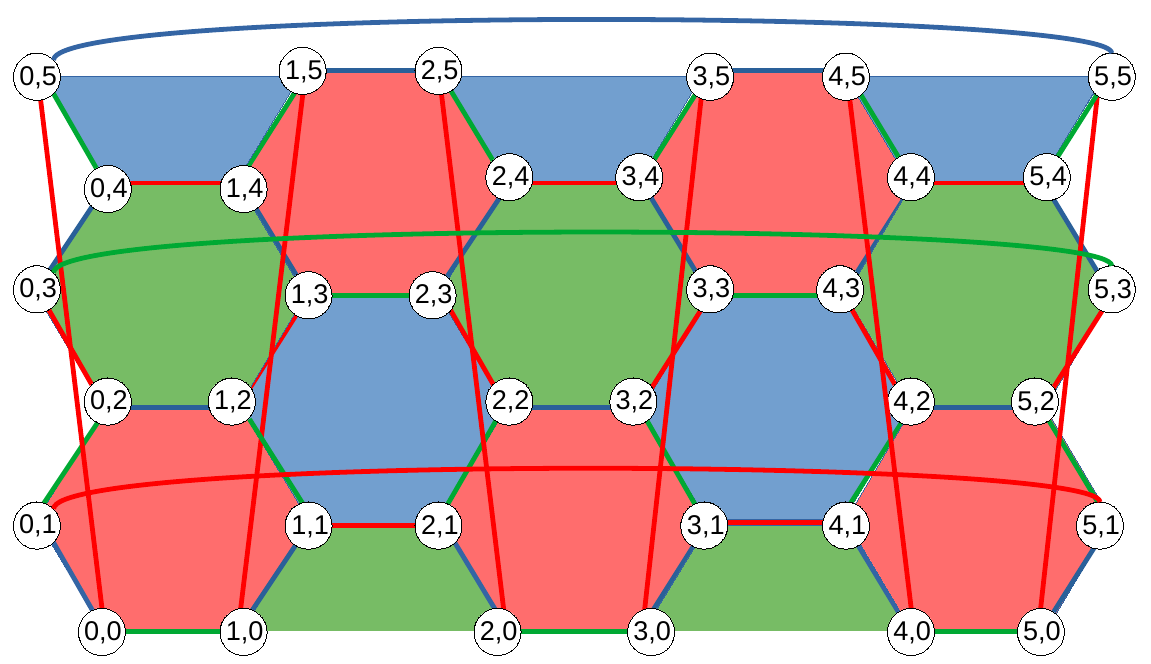}
    \caption{The toric hexagonal lattice used for the CSS honeycomb code \cite{davydova2023floquet} and our qudit generalization.}
    \label{fig:hex lattice}
\end{figure}
As an example of a dynamical code, in Fig.~\ref{fig:examples}(c) we considered the 2D CSS honeycomb code \cite{davydova2023floquet}. 
This code was introduced for qubits, but we observe that it can be straightforwardly generalized to qudits, as follows. 
The data qudits live on a toric hexagonal lattice, as shown in Fig.~\ref{fig:hex lattice}. The faces are $3$-colored red, blue, and green, and the edges are colored correspondingly: red edges have red faces at each endpoint, etc. 
In the notation of \cite{davydova2023floquet} for the qubit case, the period-6 measurement schedule can be written out as $gZZ, bXX, rZZ, gXX, bZZ, rXX$, where $gZZ$ means to measure $Z\otimes Z$ on the green edges, etc. 
For the qudit case, we simply replace $Z\otimes Z$ measurements with $Z\otimes Z^\dagger$ measurements (but leave $X\otimes X$ measurements unchanged). 
This change ensures that, for example, the product of blue $X\otimes X$ checks around a green hexagon (an instantaneous stabilizer of the code) commutes with the following round of red $Z\otimes Z^\dagger$ checks.

\bibliography{refs}

\begin{thebibliography}{43}%
\makeatletter
\providecommand \@ifxundefined [1]{%
 \@ifx{#1\undefined}
}%
\providecommand \@ifnum [1]{%
 \ifnum #1\expandafter \@firstoftwo
 \else \expandafter \@secondoftwo
 \fi
}%
\providecommand \@ifx [1]{%
 \ifx #1\expandafter \@firstoftwo
 \else \expandafter \@secondoftwo
 \fi
}%
\providecommand \natexlab [1]{#1}%
\providecommand \enquote  [1]{``#1''}%
\providecommand \bibnamefont  [1]{#1}%
\providecommand \bibfnamefont [1]{#1}%
\providecommand \citenamefont [1]{#1}%
\providecommand \href@noop [0]{\@secondoftwo}%
\providecommand \href [0]{\begingroup \@sanitize@url \@href}%
\providecommand \@href[1]{\@@startlink{#1}\@@href}%
\providecommand \@@href[1]{\endgroup#1\@@endlink}%
\providecommand \@sanitize@url [0]{\catcode `\\12\catcode `\$12\catcode `\&12\catcode `\#12\catcode `\^12\catcode `\_12\catcode `\%12\relax}%
\providecommand \@@startlink[1]{}%
\providecommand \@@endlink[0]{}%
\providecommand \url  [0]{\begingroup\@sanitize@url \@url }%
\providecommand \@url [1]{\endgroup\@href {#1}{\urlprefix }}%
\providecommand \urlprefix  [0]{URL }%
\providecommand \Eprint [0]{\href }%
\providecommand \doibase [0]{https://doi.org/}%
\providecommand \selectlanguage [0]{\@gobble}%
\providecommand \bibinfo  [0]{\@secondoftwo}%
\providecommand \bibfield  [0]{\@secondoftwo}%
\providecommand \translation [1]{[#1]}%
\providecommand \BibitemOpen [0]{}%
\providecommand \bibitemStop [0]{}%
\providecommand \bibitemNoStop [0]{.\EOS\space}%
\providecommand \EOS [0]{\spacefactor3000\relax}%
\providecommand \BibitemShut  [1]{\csname bibitem#1\endcsname}%
\let\auto@bib@innerbib\@empty
\bibitem [{\citenamefont {Raussendorf}\ and\ \citenamefont {Briegel}(2001)}]{MBQC}%
  \BibitemOpen
  \bibfield  {author} {\bibinfo {author} {\bibfnamefont {R.}~\bibnamefont {Raussendorf}}\ and\ \bibinfo {author} {\bibfnamefont {H.~J.}\ \bibnamefont {Briegel}},\ }\bibfield  {title} {\bibinfo {title} {A one-way quantum computer},\ }\href {https://doi.org/10.1103/PhysRevLett.86.5188} {\bibfield  {journal} {\bibinfo  {journal} {Phys. Rev. Lett.}\ }\textbf {\bibinfo {volume} {86}},\ \bibinfo {pages} {5188} (\bibinfo {year} {2001})}\BibitemShut {NoStop}%
\bibitem [{\citenamefont {Raussendorf}\ \emph {et~al.}(2003)\citenamefont {Raussendorf}, \citenamefont {Browne},\ and\ \citenamefont {Briegel}}]{raussendorf2003measurement}%
  \BibitemOpen
  \bibfield  {author} {\bibinfo {author} {\bibfnamefont {R.}~\bibnamefont {Raussendorf}}, \bibinfo {author} {\bibfnamefont {D.~E.}\ \bibnamefont {Browne}},\ and\ \bibinfo {author} {\bibfnamefont {H.~J.}\ \bibnamefont {Briegel}},\ }\bibfield  {title} {\bibinfo {title} {Measurement-based quantum computation on cluster states},\ }\href@noop {} {\bibfield  {journal} {\bibinfo  {journal} {Physical review A}\ }\textbf {\bibinfo {volume} {68}},\ \bibinfo {pages} {022312} (\bibinfo {year} {2003})}\BibitemShut {NoStop}%
\bibitem [{\citenamefont {Bartolucci}\ \emph {et~al.}(2023)\citenamefont {Bartolucci}, \citenamefont {Birchall}, \citenamefont {Bombín}, \citenamefont {Cable}, \citenamefont {Dawson}, \citenamefont {Gimeno-Segovia}, \citenamefont {Johnston}, \citenamefont {Kieling}, \citenamefont {Nickerson}, \citenamefont {Pant}, \citenamefont {Pastawski}, \citenamefont {Rudolph},\ and\ \citenamefont {Sparrow}}]{Bartolucci2023}%
  \BibitemOpen
  \bibfield  {author} {\bibinfo {author} {\bibfnamefont {S.}~\bibnamefont {Bartolucci}}, \bibinfo {author} {\bibfnamefont {P.}~\bibnamefont {Birchall}}, \bibinfo {author} {\bibfnamefont {H.}~\bibnamefont {Bombín}}, \bibinfo {author} {\bibfnamefont {H.}~\bibnamefont {Cable}}, \bibinfo {author} {\bibfnamefont {C.}~\bibnamefont {Dawson}}, \bibinfo {author} {\bibfnamefont {M.}~\bibnamefont {Gimeno-Segovia}}, \bibinfo {author} {\bibfnamefont {E.}~\bibnamefont {Johnston}}, \bibinfo {author} {\bibfnamefont {K.}~\bibnamefont {Kieling}}, \bibinfo {author} {\bibfnamefont {N.}~\bibnamefont {Nickerson}}, \bibinfo {author} {\bibfnamefont {M.}~\bibnamefont {Pant}}, \bibinfo {author} {\bibfnamefont {F.}~\bibnamefont {Pastawski}}, \bibinfo {author} {\bibfnamefont {T.}~\bibnamefont {Rudolph}},\ and\ \bibinfo {author} {\bibfnamefont {C.}~\bibnamefont {Sparrow}},\ }\bibfield  {title} {\bibinfo {title} {Fusion-based quantum computation},\ }\href {https://doi.org/10.1038/s41467-023-36493-1} {\bibfield  {journal} {\bibinfo
  {journal} {Nature Communications}\ }\textbf {\bibinfo {volume} {14}},\ \bibinfo {pages} {912} (\bibinfo {year} {2023})}\BibitemShut {NoStop}%
\bibitem [{\citenamefont {Bombin}\ \emph {et~al.}(2021)\citenamefont {Bombin}, \citenamefont {Kim}, \citenamefont {Litinski}, \citenamefont {Nickerson}, \citenamefont {Pant}, \citenamefont {Pastawski}, \citenamefont {Roberts},\ and\ \citenamefont {Rudolph}}]{bombin2021interleaving}%
  \BibitemOpen
  \bibfield  {author} {\bibinfo {author} {\bibfnamefont {H.}~\bibnamefont {Bombin}}, \bibinfo {author} {\bibfnamefont {I.~H.}\ \bibnamefont {Kim}}, \bibinfo {author} {\bibfnamefont {D.}~\bibnamefont {Litinski}}, \bibinfo {author} {\bibfnamefont {N.}~\bibnamefont {Nickerson}}, \bibinfo {author} {\bibfnamefont {M.}~\bibnamefont {Pant}}, \bibinfo {author} {\bibfnamefont {F.}~\bibnamefont {Pastawski}}, \bibinfo {author} {\bibfnamefont {S.}~\bibnamefont {Roberts}},\ and\ \bibinfo {author} {\bibfnamefont {T.}~\bibnamefont {Rudolph}},\ }\bibfield  {title} {\bibinfo {title} {Interleaving: Modular architectures for fault-tolerant photonic quantum computing},\ }\href@noop {} {\bibfield  {journal} {\bibinfo  {journal} {arXiv preprint arXiv:2103.08612}\ } (\bibinfo {year} {2021})}\BibitemShut {NoStop}%
\bibitem [{\citenamefont {Bombin}\ \emph {et~al.}(2023)\citenamefont {Bombin}, \citenamefont {Dawson}, \citenamefont {Farrelly}, \citenamefont {Liu}, \citenamefont {Nickerson}, \citenamefont {Pant}, \citenamefont {Pastawski},\ and\ \citenamefont {Roberts}}]{bombin2023faulttolerantcomplexes}%
  \BibitemOpen
  \bibfield  {author} {\bibinfo {author} {\bibfnamefont {H.}~\bibnamefont {Bombin}}, \bibinfo {author} {\bibfnamefont {C.}~\bibnamefont {Dawson}}, \bibinfo {author} {\bibfnamefont {T.}~\bibnamefont {Farrelly}}, \bibinfo {author} {\bibfnamefont {Y.}~\bibnamefont {Liu}}, \bibinfo {author} {\bibfnamefont {N.}~\bibnamefont {Nickerson}}, \bibinfo {author} {\bibfnamefont {M.}~\bibnamefont {Pant}}, \bibinfo {author} {\bibfnamefont {F.}~\bibnamefont {Pastawski}},\ and\ \bibinfo {author} {\bibfnamefont {S.}~\bibnamefont {Roberts}},\ }\href {https://arxiv.org/abs/2308.07844} {\bibinfo {title} {Fault-tolerant complexes}} (\bibinfo {year} {2023}),\ \Eprint {https://arxiv.org/abs/2308.07844} {arXiv:2308.07844 [quant-ph]} \BibitemShut {NoStop}%
\bibitem [{\citenamefont {Broadbent}\ \emph {et~al.}(2009)\citenamefont {Broadbent}, \citenamefont {Fitzsimons},\ and\ \citenamefont {Kashefi}}]{Broadbent_2009}%
  \BibitemOpen
  \bibfield  {author} {\bibinfo {author} {\bibfnamefont {A.}~\bibnamefont {Broadbent}}, \bibinfo {author} {\bibfnamefont {J.}~\bibnamefont {Fitzsimons}},\ and\ \bibinfo {author} {\bibfnamefont {E.}~\bibnamefont {Kashefi}},\ }\bibfield  {title} {\bibinfo {title} {Universal blind quantum computation},\ }in\ \href {https://doi.org/10.1109/focs.2009.36} {\emph {\bibinfo {booktitle} {2009 50th Annual IEEE Symposium on Foundations of Computer Science}}}\ (\bibinfo  {publisher} {IEEE},\ \bibinfo {year} {2009})\ p.\ \bibinfo {pages} {517–526}\BibitemShut {NoStop}%
\bibitem [{\citenamefont {Morimae}\ and\ \citenamefont {Fujii}(2012)}]{morimae2012blind}%
  \BibitemOpen
  \bibfield  {author} {\bibinfo {author} {\bibfnamefont {T.}~\bibnamefont {Morimae}}\ and\ \bibinfo {author} {\bibfnamefont {K.}~\bibnamefont {Fujii}},\ }\bibfield  {title} {\bibinfo {title} {Blind topological measurement-based quantum computation},\ }\href@noop {} {\bibfield  {journal} {\bibinfo  {journal} {Nature communications}\ }\textbf {\bibinfo {volume} {3}},\ \bibinfo {pages} {1036} (\bibinfo {year} {2012})}\BibitemShut {NoStop}%
\bibitem [{\citenamefont {Hayashi}\ and\ \citenamefont {Morimae}(2015)}]{hayashi2015verifiable}%
  \BibitemOpen
  \bibfield  {author} {\bibinfo {author} {\bibfnamefont {M.}~\bibnamefont {Hayashi}}\ and\ \bibinfo {author} {\bibfnamefont {T.}~\bibnamefont {Morimae}},\ }\bibfield  {title} {\bibinfo {title} {Verifiable measurement-only blind quantum computing with stabilizer testing},\ }\href@noop {} {\bibfield  {journal} {\bibinfo  {journal} {Physical review letters}\ }\textbf {\bibinfo {volume} {115}},\ \bibinfo {pages} {220502} (\bibinfo {year} {2015})}\BibitemShut {NoStop}%
\bibitem [{\citenamefont {Fitzsimons}\ and\ \citenamefont {Kashefi}(2017)}]{fitzsimons2017unconditionally}%
  \BibitemOpen
  \bibfield  {author} {\bibinfo {author} {\bibfnamefont {J.~F.}\ \bibnamefont {Fitzsimons}}\ and\ \bibinfo {author} {\bibfnamefont {E.}~\bibnamefont {Kashefi}},\ }\bibfield  {title} {\bibinfo {title} {Unconditionally verifiable blind quantum computation},\ }\href@noop {} {\bibfield  {journal} {\bibinfo  {journal} {Physical Review A}\ }\textbf {\bibinfo {volume} {96}},\ \bibinfo {pages} {012303} (\bibinfo {year} {2017})}\BibitemShut {NoStop}%
\bibitem [{\citenamefont {Raussendorf}\ \emph {et~al.}(2005)\citenamefont {Raussendorf}, \citenamefont {Bravyi},\ and\ \citenamefont {Harrington}}]{Raussendorf2005}%
  \BibitemOpen
  \bibfield  {author} {\bibinfo {author} {\bibfnamefont {R.}~\bibnamefont {Raussendorf}}, \bibinfo {author} {\bibfnamefont {S.}~\bibnamefont {Bravyi}},\ and\ \bibinfo {author} {\bibfnamefont {J.}~\bibnamefont {Harrington}},\ }\bibfield  {title} {\bibinfo {title} {Long-range quantum entanglement in noisy cluster states},\ }\href {https://doi.org/10.1103/PhysRevA.71.062313} {\bibfield  {journal} {\bibinfo  {journal} {Phys. Rev. A}\ }\textbf {\bibinfo {volume} {71}},\ \bibinfo {pages} {062313} (\bibinfo {year} {2005})}\BibitemShut {NoStop}%
\bibitem [{\citenamefont {Raussendorf}\ \emph {et~al.}(2006)\citenamefont {Raussendorf}, \citenamefont {Harrington},\ and\ \citenamefont {Goyal}}]{RAUSSENDORF20062242}%
  \BibitemOpen
  \bibfield  {author} {\bibinfo {author} {\bibfnamefont {R.}~\bibnamefont {Raussendorf}}, \bibinfo {author} {\bibfnamefont {J.}~\bibnamefont {Harrington}},\ and\ \bibinfo {author} {\bibfnamefont {K.}~\bibnamefont {Goyal}},\ }\bibfield  {title} {\bibinfo {title} {A fault-tolerant one-way quantum computer},\ }\href {https://doi.org/https://doi.org/10.1016/j.aop.2006.01.012} {\bibfield  {journal} {\bibinfo  {journal} {Annals of Physics}\ }\textbf {\bibinfo {volume} {321}},\ \bibinfo {pages} {2242} (\bibinfo {year} {2006})}\BibitemShut {NoStop}%
\bibitem [{\citenamefont {Raussendorf}\ and\ \citenamefont {Harrington}(2007)}]{Raussendorf2007}%
  \BibitemOpen
  \bibfield  {author} {\bibinfo {author} {\bibfnamefont {R.}~\bibnamefont {Raussendorf}}\ and\ \bibinfo {author} {\bibfnamefont {J.}~\bibnamefont {Harrington}},\ }\bibfield  {title} {\bibinfo {title} {Fault-tolerant quantum computation with high threshold in two dimensions},\ }\href {https://doi.org/10.1103/PhysRevLett.98.190504} {\bibfield  {journal} {\bibinfo  {journal} {Phys. Rev. Lett.}\ }\textbf {\bibinfo {volume} {98}},\ \bibinfo {pages} {190504} (\bibinfo {year} {2007})}\BibitemShut {NoStop}%
\bibitem [{\citenamefont {Raussendorf}\ \emph {et~al.}(2007)\citenamefont {Raussendorf}, \citenamefont {Harrington},\ and\ \citenamefont {Goyal}}]{Raussendorf_topological}%
  \BibitemOpen
  \bibfield  {author} {\bibinfo {author} {\bibfnamefont {R.}~\bibnamefont {Raussendorf}}, \bibinfo {author} {\bibfnamefont {J.}~\bibnamefont {Harrington}},\ and\ \bibinfo {author} {\bibfnamefont {K.}~\bibnamefont {Goyal}},\ }\bibfield  {title} {\bibinfo {title} {Topological fault-tolerance in cluster state quantum computation},\ }\href {https://doi.org/10.1088/1367-2630/9/6/199} {\bibfield  {journal} {\bibinfo  {journal} {New Journal of Physics}\ }\textbf {\bibinfo {volume} {9}},\ \bibinfo {pages} {199} (\bibinfo {year} {2007})}\BibitemShut {NoStop}%
\bibitem [{\citenamefont {Bolt}\ \emph {et~al.}(2016)\citenamefont {Bolt}, \citenamefont {Duclos-Cianci}, \citenamefont {Poulin},\ and\ \citenamefont {Stace}}]{bolt2016foliated}%
  \BibitemOpen
  \bibfield  {author} {\bibinfo {author} {\bibfnamefont {A.}~\bibnamefont {Bolt}}, \bibinfo {author} {\bibfnamefont {G.}~\bibnamefont {Duclos-Cianci}}, \bibinfo {author} {\bibfnamefont {D.}~\bibnamefont {Poulin}},\ and\ \bibinfo {author} {\bibfnamefont {T.}~\bibnamefont {Stace}},\ }\bibfield  {title} {\bibinfo {title} {Foliated quantum error-correcting codes},\ }\href@noop {} {\bibfield  {journal} {\bibinfo  {journal} {Physical review letters}\ }\textbf {\bibinfo {volume} {117}},\ \bibinfo {pages} {070501} (\bibinfo {year} {2016})}\BibitemShut {NoStop}%
\bibitem [{\citenamefont {Brown}\ and\ \citenamefont {Roberts}(2020)}]{brown2020universal}%
  \BibitemOpen
  \bibfield  {author} {\bibinfo {author} {\bibfnamefont {B.~J.}\ \bibnamefont {Brown}}\ and\ \bibinfo {author} {\bibfnamefont {S.}~\bibnamefont {Roberts}},\ }\bibfield  {title} {\bibinfo {title} {Universal fault-tolerant measurement-based quantum computation},\ }\href@noop {} {\bibfield  {journal} {\bibinfo  {journal} {Physical Review Research}\ }\textbf {\bibinfo {volume} {2}},\ \bibinfo {pages} {033305} (\bibinfo {year} {2020})}\BibitemShut {NoStop}%
\bibitem [{\citenamefont {Andriyanova}\ \emph {et~al.}(2012)\citenamefont {Andriyanova}, \citenamefont {Maurice},\ and\ \citenamefont {Tillich}}]{andriyanova2012new}%
  \BibitemOpen
  \bibfield  {author} {\bibinfo {author} {\bibfnamefont {I.}~\bibnamefont {Andriyanova}}, \bibinfo {author} {\bibfnamefont {D.}~\bibnamefont {Maurice}},\ and\ \bibinfo {author} {\bibfnamefont {J.-P.}\ \bibnamefont {Tillich}},\ }\bibfield  {title} {\bibinfo {title} {New constructions of css codes obtained by moving to higher alphabets},\ }\href@noop {} {\bibfield  {journal} {\bibinfo  {journal} {arXiv preprint arXiv:1202.3338}\ } (\bibinfo {year} {2012})}\BibitemShut {NoStop}%
\bibitem [{\citenamefont {Anwar}\ \emph {et~al.}(2014)\citenamefont {Anwar}, \citenamefont {Brown}, \citenamefont {Campbell},\ and\ \citenamefont {Browne}}]{anwar2014fast}%
  \BibitemOpen
  \bibfield  {author} {\bibinfo {author} {\bibfnamefont {H.}~\bibnamefont {Anwar}}, \bibinfo {author} {\bibfnamefont {B.~J.}\ \bibnamefont {Brown}}, \bibinfo {author} {\bibfnamefont {E.~T.}\ \bibnamefont {Campbell}},\ and\ \bibinfo {author} {\bibfnamefont {D.~E.}\ \bibnamefont {Browne}},\ }\bibfield  {title} {\bibinfo {title} {Fast decoders for qudit topological codes},\ }\href@noop {} {\bibfield  {journal} {\bibinfo  {journal} {New Journal of Physics}\ }\textbf {\bibinfo {volume} {16}},\ \bibinfo {pages} {063038} (\bibinfo {year} {2014})}\BibitemShut {NoStop}%
\bibitem [{\citenamefont {Watson}\ \emph {et~al.}(2015)\citenamefont {Watson}, \citenamefont {Anwar},\ and\ \citenamefont {Browne}}]{watson2015fast}%
  \BibitemOpen
  \bibfield  {author} {\bibinfo {author} {\bibfnamefont {F.~H.}\ \bibnamefont {Watson}}, \bibinfo {author} {\bibfnamefont {H.}~\bibnamefont {Anwar}},\ and\ \bibinfo {author} {\bibfnamefont {D.~E.}\ \bibnamefont {Browne}},\ }\bibfield  {title} {\bibinfo {title} {Fast fault-tolerant decoder for qubit and qudit surface codes},\ }\href@noop {} {\bibfield  {journal} {\bibinfo  {journal} {Physical Review A}\ }\textbf {\bibinfo {volume} {92}},\ \bibinfo {pages} {032309} (\bibinfo {year} {2015})}\BibitemShut {NoStop}%
\bibitem [{\citenamefont {Brock}\ \emph {et~al.}(2025)\citenamefont {Brock}, \citenamefont {Singh}, \citenamefont {Eickbusch}, \citenamefont {Sivak}, \citenamefont {Ding}, \citenamefont {Frunzio}, \citenamefont {Girvin},\ and\ \citenamefont {Devoret}}]{Brock_2025}%
  \BibitemOpen
  \bibfield  {author} {\bibinfo {author} {\bibfnamefont {B.~L.}\ \bibnamefont {Brock}}, \bibinfo {author} {\bibfnamefont {S.}~\bibnamefont {Singh}}, \bibinfo {author} {\bibfnamefont {A.}~\bibnamefont {Eickbusch}}, \bibinfo {author} {\bibfnamefont {V.~V.}\ \bibnamefont {Sivak}}, \bibinfo {author} {\bibfnamefont {A.~Z.}\ \bibnamefont {Ding}}, \bibinfo {author} {\bibfnamefont {L.}~\bibnamefont {Frunzio}}, \bibinfo {author} {\bibfnamefont {S.~M.}\ \bibnamefont {Girvin}},\ and\ \bibinfo {author} {\bibfnamefont {M.~H.}\ \bibnamefont {Devoret}},\ }\bibfield  {title} {\bibinfo {title} {Quantum error correction of qudits beyond break-even},\ }\href {https://doi.org/10.1038/s41586-025-08899-y} {\bibfield  {journal} {\bibinfo  {journal} {Nature}\ }\textbf {\bibinfo {volume} {641}},\ \bibinfo {pages} {612–618} (\bibinfo {year} {2025})}\BibitemShut {NoStop}%
\bibitem [{\citenamefont {Ringbauer}\ \emph {et~al.}(2022)\citenamefont {Ringbauer}, \citenamefont {Meth}, \citenamefont {Postler}, \citenamefont {Stricker}, \citenamefont {Blatt}, \citenamefont {Schindler},\ and\ \citenamefont {Monz}}]{Ringbauer2022}%
  \BibitemOpen
  \bibfield  {author} {\bibinfo {author} {\bibfnamefont {M.}~\bibnamefont {Ringbauer}}, \bibinfo {author} {\bibfnamefont {M.}~\bibnamefont {Meth}}, \bibinfo {author} {\bibfnamefont {L.}~\bibnamefont {Postler}}, \bibinfo {author} {\bibfnamefont {R.}~\bibnamefont {Stricker}}, \bibinfo {author} {\bibfnamefont {R.}~\bibnamefont {Blatt}}, \bibinfo {author} {\bibfnamefont {P.}~\bibnamefont {Schindler}},\ and\ \bibinfo {author} {\bibfnamefont {T.}~\bibnamefont {Monz}},\ }\bibfield  {title} {\bibinfo {title} {A universal qudit quantum processor with trapped ions},\ }\href {https://doi.org/10.1038/s41567-022-01658-0} {\bibfield  {journal} {\bibinfo  {journal} {Nature Physics}\ }\textbf {\bibinfo {volume} {18}},\ \bibinfo {pages} {1053} (\bibinfo {year} {2022})}\BibitemShut {NoStop}%
\bibitem [{\citenamefont {Low}\ \emph {et~al.}(2025)\citenamefont {Low}, \citenamefont {White},\ and\ \citenamefont {Senko}}]{Low2025}%
  \BibitemOpen
  \bibfield  {author} {\bibinfo {author} {\bibfnamefont {P.~J.}\ \bibnamefont {Low}}, \bibinfo {author} {\bibfnamefont {B.}~\bibnamefont {White}},\ and\ \bibinfo {author} {\bibfnamefont {C.}~\bibnamefont {Senko}},\ }\bibfield  {title} {\bibinfo {title} {Control and readout of a 13-level trapped ion qudit},\ }\href {https://doi.org/10.1038/s41534-025-01031-y} {\bibfield  {journal} {\bibinfo  {journal} {npj Quantum Information}\ }\textbf {\bibinfo {volume} {11}},\ \bibinfo {pages} {85} (\bibinfo {year} {2025})}\BibitemShut {NoStop}%
\bibitem [{\citenamefont {Zhou}\ \emph {et~al.}(2003)\citenamefont {Zhou}, \citenamefont {Zeng}, \citenamefont {Xu},\ and\ \citenamefont {Sun}}]{Zhou_2003}%
  \BibitemOpen
  \bibfield  {author} {\bibinfo {author} {\bibfnamefont {D.~L.}\ \bibnamefont {Zhou}}, \bibinfo {author} {\bibfnamefont {B.}~\bibnamefont {Zeng}}, \bibinfo {author} {\bibfnamefont {Z.}~\bibnamefont {Xu}},\ and\ \bibinfo {author} {\bibfnamefont {C.~P.}\ \bibnamefont {Sun}},\ }\bibfield  {title} {\bibinfo {title} {Quantum computation based on<i>d</i>-level cluster state},\ }\bibfield  {journal} {\bibinfo  {journal} {Physical Review A}\ }\textbf {\bibinfo {volume} {68}},\ \href {https://doi.org/10.1103/physreva.68.062303} {10.1103/physreva.68.062303} (\bibinfo {year} {2003})\BibitemShut {NoStop}%
\bibitem [{\citenamefont {Booth}\ \emph {et~al.}(2023)\citenamefont {Booth}, \citenamefont {Kissinger}, \citenamefont {Markham}, \citenamefont {Meignant},\ and\ \citenamefont {Perdrix}}]{booth2023outcome}%
  \BibitemOpen
  \bibfield  {author} {\bibinfo {author} {\bibfnamefont {R.~I.}\ \bibnamefont {Booth}}, \bibinfo {author} {\bibfnamefont {A.}~\bibnamefont {Kissinger}}, \bibinfo {author} {\bibfnamefont {D.}~\bibnamefont {Markham}}, \bibinfo {author} {\bibfnamefont {C.}~\bibnamefont {Meignant}},\ and\ \bibinfo {author} {\bibfnamefont {S.}~\bibnamefont {Perdrix}},\ }\bibfield  {title} {\bibinfo {title} {Outcome determinism in measurement-based quantum computation with qudits},\ }\href@noop {} {\bibfield  {journal} {\bibinfo  {journal} {Journal of Physics A: Mathematical and Theoretical}\ }\textbf {\bibinfo {volume} {56}},\ \bibinfo {pages} {115303} (\bibinfo {year} {2023})}\BibitemShut {NoStop}%
\bibitem [{\citenamefont {Romanova}\ and\ \citenamefont {Dür}(2026)}]{Romanova_2026}%
  \BibitemOpen
  \bibfield  {author} {\bibinfo {author} {\bibfnamefont {A.}~\bibnamefont {Romanova}}\ and\ \bibinfo {author} {\bibfnamefont {W.}~\bibnamefont {Dür}},\ }\bibfield  {title} {\bibinfo {title} {Measurement-based quantum computing with qudit stabilizer states},\ }\href {https://doi.org/10.1088/2058-9565/ae3b6f} {\bibfield  {journal} {\bibinfo  {journal} {Quantum Science and Technology}\ }\textbf {\bibinfo {volume} {11}},\ \bibinfo {pages} {015054} (\bibinfo {year} {2026})}\BibitemShut {NoStop}%
\bibitem [{\citenamefont {Chau}(1997)}]{chau1997five}%
  \BibitemOpen
  \bibfield  {author} {\bibinfo {author} {\bibfnamefont {H.}~\bibnamefont {Chau}},\ }\bibfield  {title} {\bibinfo {title} {Five quantum register error correction code for higher spin systems},\ }\href@noop {} {\bibfield  {journal} {\bibinfo  {journal} {Physical Review A}\ }\textbf {\bibinfo {volume} {56}},\ \bibinfo {pages} {R1} (\bibinfo {year} {1997})}\BibitemShut {NoStop}%
\bibitem [{\citenamefont {Davydova}\ \emph {et~al.}(2023)\citenamefont {Davydova}, \citenamefont {Tantivasadakarn},\ and\ \citenamefont {Balasubramanian}}]{davydova2023floquet}%
  \BibitemOpen
  \bibfield  {author} {\bibinfo {author} {\bibfnamefont {M.}~\bibnamefont {Davydova}}, \bibinfo {author} {\bibfnamefont {N.}~\bibnamefont {Tantivasadakarn}},\ and\ \bibinfo {author} {\bibfnamefont {S.}~\bibnamefont {Balasubramanian}},\ }\bibfield  {title} {\bibinfo {title} {Floquet codes without parent subsystem codes},\ }\href@noop {} {\bibfield  {journal} {\bibinfo  {journal} {PRX Quantum}\ }\textbf {\bibinfo {volume} {4}},\ \bibinfo {pages} {020341} (\bibinfo {year} {2023})}\BibitemShut {NoStop}%
\bibitem [{\citenamefont {Helwig}(2013)}]{helwig2013absolutelymaximallyentangledqudit}%
  \BibitemOpen
  \bibfield  {author} {\bibinfo {author} {\bibfnamefont {W.}~\bibnamefont {Helwig}},\ }\href {https://arxiv.org/abs/1306.2879} {\bibinfo {title} {Absolutely maximally entangled qudit graph states}} (\bibinfo {year} {2013}),\ \Eprint {https://arxiv.org/abs/1306.2879} {arXiv:1306.2879 [quant-ph]} \BibitemShut {NoStop}%
\bibitem [{\citenamefont {Gottesman}(1998)}]{gottesman1998fault}%
  \BibitemOpen
  \bibfield  {author} {\bibinfo {author} {\bibfnamefont {D.}~\bibnamefont {Gottesman}},\ }\bibfield  {title} {\bibinfo {title} {Fault-tolerant quantum computation with higher-dimensional systems},\ }in\ \href@noop {} {\emph {\bibinfo {booktitle} {NASA International Conference on Quantum Computing and Quantum Communications}}}\ (\bibinfo {organization} {Springer},\ \bibinfo {year} {1998})\ pp.\ \bibinfo {pages} {302--313}\BibitemShut {NoStop}%
\bibitem [{\citenamefont {Poulin}(2005)}]{poulin2005stabilizer}%
  \BibitemOpen
  \bibfield  {author} {\bibinfo {author} {\bibfnamefont {D.}~\bibnamefont {Poulin}},\ }\bibfield  {title} {\bibinfo {title} {Stabilizer formalism for operator quantum error correction},\ }\href@noop {} {\bibfield  {journal} {\bibinfo  {journal} {Physical review letters}\ }\textbf {\bibinfo {volume} {95}},\ \bibinfo {pages} {230504} (\bibinfo {year} {2005})}\BibitemShut {NoStop}%
\bibitem [{\citenamefont {Fu}\ and\ \citenamefont {Gottesman}(2025)}]{fu2025error}%
  \BibitemOpen
  \bibfield  {author} {\bibinfo {author} {\bibfnamefont {E.~X.}\ \bibnamefont {Fu}}\ and\ \bibinfo {author} {\bibfnamefont {D.}~\bibnamefont {Gottesman}},\ }\bibfield  {title} {\bibinfo {title} {Error correction in dynamical codes},\ }\href@noop {} {\bibfield  {journal} {\bibinfo  {journal} {Quantum}\ }\textbf {\bibinfo {volume} {9}},\ \bibinfo {pages} {1886} (\bibinfo {year} {2025})}\BibitemShut {NoStop}%
\bibitem [{\citenamefont {Bullock}\ and\ \citenamefont {Brennen}(2007)}]{Bullock_2007}%
  \BibitemOpen
  \bibfield  {author} {\bibinfo {author} {\bibfnamefont {S.~S.}\ \bibnamefont {Bullock}}\ and\ \bibinfo {author} {\bibfnamefont {G.~K.}\ \bibnamefont {Brennen}},\ }\bibfield  {title} {\bibinfo {title} {Qudit surface codes and gauge theory with finite cyclic groups},\ }\href {https://doi.org/10.1088/1751-8113/40/13/013} {\bibfield  {journal} {\bibinfo  {journal} {Journal of Physics A: Mathematical and Theoretical}\ }\textbf {\bibinfo {volume} {40}},\ \bibinfo {pages} {3481–3505} (\bibinfo {year} {2007})}\BibitemShut {NoStop}%
\bibitem [{\citenamefont {Gimeno-Segovia}\ \emph {et~al.}(2019)\citenamefont {Gimeno-Segovia}, \citenamefont {Rudolph},\ and\ \citenamefont {Economou}}]{gimeno2019deterministic}%
  \BibitemOpen
  \bibfield  {author} {\bibinfo {author} {\bibfnamefont {M.}~\bibnamefont {Gimeno-Segovia}}, \bibinfo {author} {\bibfnamefont {T.}~\bibnamefont {Rudolph}},\ and\ \bibinfo {author} {\bibfnamefont {S.~E.}\ \bibnamefont {Economou}},\ }\bibfield  {title} {\bibinfo {title} {Deterministic generation of large-scale entangled photonic cluster state from interacting solid state emitters},\ }\href@noop {} {\bibfield  {journal} {\bibinfo  {journal} {Physical review letters}\ }\textbf {\bibinfo {volume} {123}},\ \bibinfo {pages} {070501} (\bibinfo {year} {2019})}\BibitemShut {NoStop}%
\bibitem [{\citenamefont {Raissi}\ \emph {et~al.}(2024)\citenamefont {Raissi}, \citenamefont {Barnes},\ and\ \citenamefont {Economou}}]{raissi2024deterministic}%
  \BibitemOpen
  \bibfield  {author} {\bibinfo {author} {\bibfnamefont {Z.}~\bibnamefont {Raissi}}, \bibinfo {author} {\bibfnamefont {E.}~\bibnamefont {Barnes}},\ and\ \bibinfo {author} {\bibfnamefont {S.~E.}\ \bibnamefont {Economou}},\ }\bibfield  {title} {\bibinfo {title} {Deterministic generation of qudit photonic graph states from quantum emitters},\ }\href@noop {} {\bibfield  {journal} {\bibinfo  {journal} {PRX Quantum}\ }\textbf {\bibinfo {volume} {5}},\ \bibinfo {pages} {020346} (\bibinfo {year} {2024})}\BibitemShut {NoStop}%
\bibitem [{\citenamefont {\"Ust\"un}\ and\ \citenamefont {Devitt}(2026)}]{comparing_schemes}%
  \BibitemOpen
  \bibfield  {author} {\bibinfo {author} {\bibfnamefont {G.}~\bibnamefont {\"Ust\"un}}\ and\ \bibinfo {author} {\bibfnamefont {S.~J.}\ \bibnamefont {Devitt}},\ }\bibfield  {title} {\bibinfo {title} {Comparing schemes for creating qudit graph states from 16- and 128-dimensional hilbert space using donors in silicon},\ }\href {https://doi.org/10.1103/dcvl-sdxq} {\bibfield  {journal} {\bibinfo  {journal} {Phys. Rev. Res.}\ }\textbf {\bibinfo {volume} {8}},\ \bibinfo {pages} {013343} (\bibinfo {year} {2026})}\BibitemShut {NoStop}%
\bibitem [{\citenamefont {Lee}\ and\ \citenamefont {Jeong}(2023)}]{lee2023graph}%
  \BibitemOpen
  \bibfield  {author} {\bibinfo {author} {\bibfnamefont {S.-H.}\ \bibnamefont {Lee}}\ and\ \bibinfo {author} {\bibfnamefont {H.}~\bibnamefont {Jeong}},\ }\bibfield  {title} {\bibinfo {title} {Graph-theoretical optimization of fusion-based graph state generation},\ }\href@noop {} {\bibfield  {journal} {\bibinfo  {journal} {Quantum}\ }\textbf {\bibinfo {volume} {7}},\ \bibinfo {pages} {1212} (\bibinfo {year} {2023})}\BibitemShut {NoStop}%
\bibitem [{\citenamefont {Pankovich}\ \emph {et~al.}(2024)\citenamefont {Pankovich}, \citenamefont {Neville}, \citenamefont {Kan}, \citenamefont {Omkar}, \citenamefont {Wan},\ and\ \citenamefont {Br{\'a}dler}}]{pankovich2024flexible}%
  \BibitemOpen
  \bibfield  {author} {\bibinfo {author} {\bibfnamefont {B.}~\bibnamefont {Pankovich}}, \bibinfo {author} {\bibfnamefont {A.}~\bibnamefont {Neville}}, \bibinfo {author} {\bibfnamefont {A.}~\bibnamefont {Kan}}, \bibinfo {author} {\bibfnamefont {S.}~\bibnamefont {Omkar}}, \bibinfo {author} {\bibfnamefont {K.~H.}\ \bibnamefont {Wan}},\ and\ \bibinfo {author} {\bibfnamefont {K.}~\bibnamefont {Br{\'a}dler}},\ }\bibfield  {title} {\bibinfo {title} {Flexible entangled-state generation in linear optics},\ }\href@noop {} {\bibfield  {journal} {\bibinfo  {journal} {Physical Review A}\ }\textbf {\bibinfo {volume} {110}},\ \bibinfo {pages} {032402} (\bibinfo {year} {2024})}\BibitemShut {NoStop}%
\bibitem [{\citenamefont {L{\"o}bl}\ \emph {et~al.}(2025)\citenamefont {L{\"o}bl}, \citenamefont {Pettersson}, \citenamefont {Paesani},\ and\ \citenamefont {S{\o}rensen}}]{lobl2025transforming}%
  \BibitemOpen
  \bibfield  {author} {\bibinfo {author} {\bibfnamefont {M.~C.}\ \bibnamefont {L{\"o}bl}}, \bibinfo {author} {\bibfnamefont {L.~A.}\ \bibnamefont {Pettersson}}, \bibinfo {author} {\bibfnamefont {S.}~\bibnamefont {Paesani}},\ and\ \bibinfo {author} {\bibfnamefont {A.~S.}\ \bibnamefont {S{\o}rensen}},\ }\bibfield  {title} {\bibinfo {title} {Transforming graph states via bell state measurements},\ }\href@noop {} {\bibfield  {journal} {\bibinfo  {journal} {Quantum}\ }\textbf {\bibinfo {volume} {9}},\ \bibinfo {pages} {1795} (\bibinfo {year} {2025})}\BibitemShut {NoStop}%
\bibitem [{\citenamefont {Browne}\ and\ \citenamefont {Rudolph}(2005)}]{browne_rudolph}%
  \BibitemOpen
  \bibfield  {author} {\bibinfo {author} {\bibfnamefont {D.~E.}\ \bibnamefont {Browne}}\ and\ \bibinfo {author} {\bibfnamefont {T.}~\bibnamefont {Rudolph}},\ }\bibfield  {title} {\bibinfo {title} {Resource-efficient linear optical quantum computation},\ }\href {https://doi.org/10.1103/PhysRevLett.95.010501} {\bibfield  {journal} {\bibinfo  {journal} {Phys. Rev. Lett.}\ }\textbf {\bibinfo {volume} {95}},\ \bibinfo {pages} {010501} (\bibinfo {year} {2005})}\BibitemShut {NoStop}%
\bibitem [{\citenamefont {{\"U}st{\"u}n}\ \emph {et~al.}(2025)\citenamefont {{\"U}st{\"u}n}, \citenamefont {Rieffel}, \citenamefont {Devitt},\ and\ \citenamefont {Saied}}]{ustun2025fusion}%
  \BibitemOpen
  \bibfield  {author} {\bibinfo {author} {\bibfnamefont {G.}~\bibnamefont {{\"U}st{\"u}n}}, \bibinfo {author} {\bibfnamefont {E.~G.}\ \bibnamefont {Rieffel}}, \bibinfo {author} {\bibfnamefont {S.~J.}\ \bibnamefont {Devitt}},\ and\ \bibinfo {author} {\bibfnamefont {J.}~\bibnamefont {Saied}},\ }\bibfield  {title} {\bibinfo {title} {Fusion for high-dimensional linear-optical quantum computing with improved success probability},\ }\href@noop {} {\bibfield  {journal} {\bibinfo  {journal} {Physical Review Applied}\ }\textbf {\bibinfo {volume} {24}},\ \bibinfo {pages} {044024} (\bibinfo {year} {2025})}\BibitemShut {NoStop}%
\bibitem [{\citenamefont {Kabir}\ \emph {et~al.}(2026)\citenamefont {Kabir}, \citenamefont {Nguyen}, \citenamefont {Ghosh}, \citenamefont {Keppens}, \citenamefont {Kiran}, \citenamefont {Kim}, \citenamefont {Huang},\ and\ \citenamefont {Sorée}}]{kabir2026sdimquditstabilizersimulator}%
  \BibitemOpen
  \bibfield  {author} {\bibinfo {author} {\bibfnamefont {A.}~\bibnamefont {Kabir}}, \bibinfo {author} {\bibfnamefont {S.}~\bibnamefont {Nguyen}}, \bibinfo {author} {\bibfnamefont {S.}~\bibnamefont {Ghosh}}, \bibinfo {author} {\bibfnamefont {J.}~\bibnamefont {Keppens}}, \bibinfo {author} {\bibfnamefont {T.}~\bibnamefont {Kiran}}, \bibinfo {author} {\bibfnamefont {I.~H.}\ \bibnamefont {Kim}}, \bibinfo {author} {\bibfnamefont {Y.}~\bibnamefont {Huang}},\ and\ \bibinfo {author} {\bibfnamefont {B.}~\bibnamefont {Sorée}},\ }\href {https://arxiv.org/abs/2511.12777} {\bibinfo {title} {Sdim: A qudit stabilizer simulator}} (\bibinfo {year} {2026}),\ \Eprint {https://arxiv.org/abs/2511.12777} {arXiv:2511.12777 [quant-ph]} \BibitemShut {NoStop}%
\bibitem [{\citenamefont {Bahramgiri}\ and\ \citenamefont {Beigi}(2006)}]{bahramgiri2006graph}%
  \BibitemOpen
  \bibfield  {author} {\bibinfo {author} {\bibfnamefont {M.}~\bibnamefont {Bahramgiri}}\ and\ \bibinfo {author} {\bibfnamefont {S.}~\bibnamefont {Beigi}},\ }\bibfield  {title} {\bibinfo {title} {Graph states under the action of local clifford group in non-binary case},\ }\href@noop {} {\bibfield  {journal} {\bibinfo  {journal} {arXiv preprint quant-ph/0610267}\ } (\bibinfo {year} {2006})}\BibitemShut {NoStop}%
\bibitem [{\citenamefont {Zhou}\ \emph {et~al.}(2000)\citenamefont {Zhou}, \citenamefont {Leung},\ and\ \citenamefont {Chuang}}]{zhou2000methodology}%
  \BibitemOpen
  \bibfield  {author} {\bibinfo {author} {\bibfnamefont {X.}~\bibnamefont {Zhou}}, \bibinfo {author} {\bibfnamefont {D.~W.}\ \bibnamefont {Leung}},\ and\ \bibinfo {author} {\bibfnamefont {I.~L.}\ \bibnamefont {Chuang}},\ }\bibfield  {title} {\bibinfo {title} {Methodology for quantum logic gate construction},\ }\href@noop {} {\bibfield  {journal} {\bibinfo  {journal} {Physical Review A}\ }\textbf {\bibinfo {volume} {62}},\ \bibinfo {pages} {052316} (\bibinfo {year} {2000})}\BibitemShut {NoStop}%
\bibitem [{\citenamefont {Nielsen}(2006)}]{nielsen2006cluster}%
  \BibitemOpen
  \bibfield  {author} {\bibinfo {author} {\bibfnamefont {M.~A.}\ \bibnamefont {Nielsen}},\ }\bibfield  {title} {\bibinfo {title} {Cluster-state quantum computation},\ }\href@noop {} {\bibfield  {journal} {\bibinfo  {journal} {Reports on Mathematical Physics}\ }\textbf {\bibinfo {volume} {57}},\ \bibinfo {pages} {147} (\bibinfo {year} {2006})}\BibitemShut {NoStop}%
\end{thebibliography}%

\newpage
\onecolumngrid
\section*{Supplementary Materials}
\setcounter{section}{0}
\setcounter{equation}{0}
\setcounter{figure}{0}
\renewcommand{\thesection}{S-\Alph{section}}
\renewcommand{\theequation}{S.\arabic{equation}}
\renewcommand{\thefigure}{S.\arabic{figure}}

Sec.~\ref{sec:non css theory} formalizes the MBQC construction for the non-CSS case. The remaining sections give less formal expositions of the main ideas and discuss examples. 

\section{Foliating Non-CSS Codes}\label{sec:non css theory}
We now formally discuss the non-CSS version of Definition~\ref{def:css graph state} above. We give further exposition of the theory in Sec.~\ref{sec:non css example} and discuss the example of the $5$-qudit perfect code in Sec.~\ref{sec:perfect}. 

In the CSS case, we had $T/2$ pairs of ``layers'' or ``time steps,'' with each pair corresponding to measurement of $Z$ checks followed by measurement of $X$ checks. 
In the non-CSS setting, where checks cannot be separated into $X$ and $Z$ type, these pairs of layers will be combined into a single \emph{measurement round}. 
We will thus refer to having $K$ rounds rather than $T$ time steps. 

We consider a Pauli-based code (possibly dynamical) involving $N$ data qudits of dimension $d$. 
For $0\leq k\leq K$, we assume we are given a $r(k)\times 2N$ generator matrix $\tilde{H}(k)$. 
(We use the tilde notation to avoid confusion of notation with the CSS case above.) 
The rows of $\tilde{H}(k)$ specify the checks measured in round $k$; specifically, recalling the expression of Paulis as $X^{p_X}Z^{p_Z}$, the first $N$ entries in a row specify $p_Z$, and the final $N$ entries specify $p_X$. 
(In the CSS case, we obtain a block diagonal matrix with the blocks encoding the separate $Z$ and $X$ checks.) 

\begin{definition}\label{def:non css}
Given a code specified by $d, N, K$, and $\tilde{H}(k)$ as above, we construct the corresponding graph state $\tilde{\Lambda}=\tilde{\Lambda}(d,N,K,\tilde{H})$: 
\begin{enumerate}
    \item We have the set of \emph{data qudit nodes} 
        $\mathcal{D} = \{(0,q,t): 1\leq q \leq N, 0\leq t \leq 2K\}$.
    For $t>0$, we add an edge from $(0,q,t-1)$ to $(0,q,t)$, with weight $(-1)^t$. 
    \item For $0\leq k \leq K$, we have the set of \emph{ancilla qudit nodes}
        $\mathcal{A}_k = \{(1, c, k): 1\leq c \leq r(t)\}$,
    with each $c$ corresponding to a check in measurement round $k$ (a row of $\tilde{H}(k)$). 
    We add the following edges: 
    \begin{enumerate}
        \item For each such node $(1,c,k)$, and each data qudit $q$ in the support of $c$, let $a=H(k)_{c,q}$ and $b=H(k)_{c,q+N}$, so that the part of the check supported on $q$ has the form $X^b Z^a$. 
        We add an edge to $(0,q,2k)$ with weight $a$ and an edge to $(0,q,2k+1)$ with weight $b$. 
        \item For \emph{every} pair of distinct checks $c_p, c_q$ measured during the same round $k$, 
        corresponding to checks $X^{p_X}Z^{p_Z}$ and $X^{q_X}Z^{q_Z}$ respectively, 
        we add an edge of weight $p_X\cdot q_Z= p_Z\cdot q_X$ between nodes $(1,c_p,k)$ and $(1,c_q,k)$. 
    \end{enumerate}
    
\end{enumerate}
The measurement-based quantum computation proceeds by measuring (one round at a time, in order) $XZ^{-p_X\cdot p_Z}$ on each ancillary qudit corresponding to a check $X^{p_X}Z^{p_Z}$, and $X$ on each data qudit.  
\end{definition}

The detectors may be constructed similarly to the CSS case. 
For example, in the stabilizer case, 
we have the following straightforward generalization of Definition~\ref{def:detector}. 
For each ancilla node $(1,c,k)$ with $k+1\leq K$ and $c=X^{p_X}Z^{p_Z}$, the associated detector is obtained by: 
\begin{equation}\label{eq:non css detector}
    D(c,k)=(XZ^{-p_X\cdot p_Z})_{(1,c,k)}(XZ^{-p_X\cdot p_Z})_{(1,c,k+1)}^{-1} \times \left(\prod_{q=0}^{N-1} X_{(0,q,2k+1)}^{(p_Z)_{q}}\right)\times \left(\prod_{q=0}^{N-1} X_{(0,q,2k+2)}^{-(p_X)_{q}}\right).
\end{equation}
Recall that in the CSS case, $Z$-type detectors starting in time step $t=2k$ involved data qudits in time step $2k+1$, and similarly $X$-type detectors starting in time step $t=2k+1$ involved data qudits in time step $2k+2$. 
The same principle applied here: the detector is supported on data qudits $(0,q,2k+1)$ because of the $Z$ part of the check, and it is supported on data qudits $(0,q,2k+2)$ because of the $X$ part of the check. 

As in the CSS case, the detectors for a non-CSS subsystem code are obtained by appropriately multiplying detectors of the form \eqref{eq:non css detector}, and the detectors for a dynamical code are similar in spirit but may involve more than two measurement rounds. 

In this work, we considered an error model in which Pauli errors occur after preparation of the graph state $\tilde{\Lambda}$.  
In the CSS case, since we measured $X$ on every qudit, only Pauli $Z$ errors were relevant. 
In the non-CSS setting, we still measure $X$ on every data qudit, so the only relevant data qudit errors are $Z$ errors. 
On the ancillas, we measure more general observables $P = XZ^{b}$ (where $b=-p_X\cdot p_Z$ as in Definition~\ref{def:non css}). Let $\ket{\phi_0}$ be the $+1$-eigenstate of $P$; since $PZ = \omega^{-1}ZP$, we have $PZ^k\ket{\phi_0} = \omega^{-k}Z^k\ket{\phi_0}$, so $Z$ cycles between the eigenstates of $P$, just as in the CSS case. In other words, to model Pauli errors on the ancillas (corresponding to measurement errors), it suffices to consider Pauli $Z$ errors. 

\section{Exposition of Qudit Graph States}\label{sec:graph states}
We briefly review essential properties of qudit graph states. For a more detailed exposition with similar notation, we refer to \cite{helwig2013absolutelymaximallyentangledqudit}. 

We first recall the $CZ$ gate, 
\begin{equation*}
    CZ = \sum_{a,b} \omega^{ab}\ketbra{ab} = \sum_a \ketbra{a}\otimes Z^a = \sum_b Z^b \otimes \ketbra{b}.
\end{equation*}
We note that $CZ$ is symmetric with respect to permutation of the qudits. 
For computations with stabilizer states, it is useful to understand how the Clifford gate $CZ$ operates on the Paulis. Of course, $CZ$ commutes with $Z\otimes I$ and $I\otimes Z$, and moreover,
\begin{equation*}
    (CZ)^k (I\otimes X) = (Z^{k}\otimes X)(CZ)^k.
\end{equation*}
The $k=1$ case is a straightforward computation using the definitions and the identity $ XZ = \omega^{-1} ZX$:
\begin{align*}
    CZ(I\otimes X) &= \sum_a \ketbra{a}\otimes (Z^a X)
    \\
    &= \sum_a \ketbra{a}\otimes \omega^a X Z^a 
    \\
    &= (Z\otimes X)\sum_a \ketbra{a}\otimes Z^a
    \\
    &= (Z\otimes X)CZ.
\end{align*}
The general case follows by repeated application of the $CZ$ gate. 


Recall the qudit graph state corresponding to a weighted graph $G$ with nodes $\{1, \dots, N\}$, edge set $E$, and weights $w_{ij}$, is given by 
\begin{equation*}
    \prod_{(i,j)\in E} (CZ_{i,j})^{w_{ij}}\ket{+}^{\otimes N}.
\end{equation*}
Since the $\ket{+}$ state is stabilized by $X$, the above implies that, for every node $i$, the graph state is stabilized by 
\begin{equation}\label{eq:qudit stab}
    X_i \prod_{j: (i,j)\in E} Z_j^{w_{ij}}.
\end{equation}
Analogously to the qubit case, these generate the stabilizer group of the graph state, and every stabilizer code is equivalent to a graph state (up to local Clifford operations) \cite{bahramgiri2006graph}. 

The form of these stabilizers directly informs the structure of qudit MBQC. In particular, to measure a stabilizer of the form $\prod_j Z_j^{p_j}$, one must only apply $CZ^{p_j}$ between the relevant qudits and an ancilla in the $\ket{+}$ state, then measure $X$ on the ancilla. This is depicted in Fig.~\ref{fig:quditgraph}. 

\begin{figure}[H]
    \centering
    \includegraphics[width=.3\linewidth]{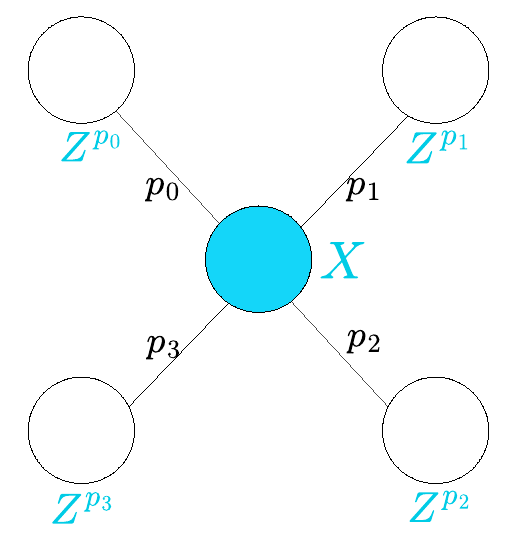}
    \caption{One of the generating stabilizers of a qudit graph state. A label $p_i$ on an edge means that we apply $(CZ)^{p_i}$ between those nodes. Since the graph state is stabilized by $Z^{p_0}\otimes\cdots\otimes Z^{p_3}\otimes X$, measuring $X$ on the marked node is equivalent to measuring $Z^{p_0}\otimes\cdots\otimes Z^{p_3}$ on the others.}
    \label{fig:quditgraph}
\end{figure}
\section{Qudit Teleportation}



As in the qubit case, the key to MBQC is a simple teleportation circuit, often referred to as \emph{one-bit teleportation} \cite{zhou2000methodology, nielsen2006cluster}. We depict the qudit case in Fig.~\ref{fig:teleportation}(a). 
We are able to teleport a state $\ket{\psi}$ to a new qudit by creating appropriate entanglement and measuring out the original qudit. 
This is how logical information is passed between layers in MBQC: when the data qudits in layer $t$ are measured, their state is teleported to the data qudits in layer $t+1$, up to a known Fourier factor (depending on the power of $CZ$ used) and a known Pauli factor (depending on the measurement result). 
In the qubit case, where $F$ is replaced with the Hadamard transform $H$, we have $H^2=I$, so the Fourier factors are automatically cancelled after two layers. 
For qudit dimension $d>2$, we have $F^2\neq I$. This is why, in Definition~\ref{def:css graph state} Part 1, we choose the edge weights (powers of $CZ$) to alternate between layers, so that we obtain cancelling factors of $F$ and $F^\dagger$. 
This is depicted in Fig.~\ref{fig:teleportation}(b). 
For completeness, we provide a quick proof of the teleportation identity: 
\begin{lemma}\label{lemma:fourier}
    We have
    \begin{equation}
        ((Z^k \ket{+})^\dagger\otimes I) CZ^{\pm 1} \left(\ket{\psi}\otimes \ket{+}\right) = \dfrac{1}{\sqrt{d}}X^{\pm k} F^{\pm 1}\ket{\psi}\label{eq:tele condensed}.
    \end{equation}
\end{lemma}
\begin{proof}
    We begin with the $k=0$ case. Writing $\ket{\psi} = \sum_a \psi_a \ket{a}$, we have
    \begin{align*}
        (\bra{+}\otimes I) CZ^{\pm 1} \left(\ket{\psi}\otimes \ket{+}\right) &=  (\bra{+}\otimes I)\frac{1}{\sqrt{d}}\sum_{a,b}  \omega^{\pm ab} \psi_a\ket{ab}
        \\&= \frac{1}{d}\sum_{a,b} \omega^{\pm ab}\psi_a \ket{b}
        \\&= \frac{1}{\sqrt{d}}F^{\pm 1}\ket{\psi}.
    \end{align*}
    The general case follows from the $k=0$ case, the relation $FZF^\dagger = X^\dagger$, and the fact that $Z$ commutes with $CZ$: 
    \begin{align*}
        ((Z^k \ket{+})^\dagger\otimes I) CZ^{\pm 1} \left(\ket{\psi}\otimes \ket{+}\right) &= (\bra{+}\otimes I) CZ^{\pm 1} \left(Z^{-k}\ket{\psi}\otimes \ket{+}\right)
        \\&= \frac{1}{\sqrt{d}}F^{\pm 1}(Z^{\mp k}\ket{\psi})
        \\&= \frac{1}{\sqrt{d}}X^{\pm k}F^{\pm 1}\ket{\psi}.
    \end{align*}

\end{proof}
\section{Exposition and examples regarding detectors}
In this section, we give a gentle review of several important concepts. We begin by discussing how a CSS stabilizer measurement is performed in MBQC and what the corresponding detector looks like. We then discuss how to utilize the same framework to measure mixed checks for non-CSS codes. 
We then discuss the cases of the toric, perfect, and CSS honeycomb codes in more detail. 
\subsection{General example: Stabilizer detector (CSS case)}\label{sec:detector general example}
We consider the example of a simple four-body $Z$ stabilizer of the form $S=Z^{p_0}\otimes Z^{p_1}\otimes Z^{p_2}\otimes Z^{p_3}$, as shown in Fig.~\ref{fig:csschecks}. 
We let $\Lambda$ be the corresponding graph state (with the understanding that the figure only depicts the parts of the graph state that are currently relevant to us). 
In Section~\ref{sec:graph states} (especially see Fig.~\ref{fig:quditgraph}), we argued that measuring $X$ on the ancilla qudit in $t=0$ causes the surrounding data qudits to be stabilized by $S$ (up to a known Pauli correction, depending on the random measurement outcome). 
In Fig.~\ref{fig:csschecks}, we present how to calculate the corresponding \emph{detector}, whose value should be $1$ in the absence of errors. 
In particular, to obtain a detector, we look for stabilizers of the graph state $\Lambda$ that commute with all measurements we will make. In our case, we will make $X$ measurements on every qudit, so we aim for a detector composed of powers of $X$. 

We begin by multiplying the stabilizers of the form \eqref{eq:qudit stab} centered at the ancillas in $t=0$ and $t=2$; motivated by \eqref{eq:detector explicit}, we use the inverse of the stabilizer at $t=2$. This gives Fig.~\ref{fig:csschecks}(a). 
So far, this intuitively matches the circuit-based setting; we are essentially taking the difference of the stabilizer measurements in two different time steps, to see if any errors have occurred in between. 
However, we need to cancel the $Z$'s on the data qudits in order to obtain a detector that will commute with the upcoming $X$ measurements. 
To do this, we multiply by (the $p_i$th powers of) the appropriate stabilizers centered on the data qudits at $t=1$, as shown in Fig.~\ref{fig:csschecks}(b). 
Due to the weights of the edges, this exactly cancels the $Z$'s on the data qudits, giving a detector of the form \eqref{eq:detector explicit}, supported on the ancillas at $t=0$ and $t=2$ and the data qudits at $t=1$. 
Since we measure $X$ on every qudit, the value of this detector will automatically be obtained during the course of the MBQC. 

A cautionary note: here we showed only the nodes directly relevant to measuring the stabilizer $S=Z^{p_0}\otimes Z^{p_1}\otimes Z^{p_2}\otimes Z^{p_3}$ and calculating its detector. 
However, if we are simultaneously measuring $X$-type stabilizers starting at $t=1$, then going from (a) to (b) in Fig.~\ref{fig:csschecks} requires one further observation. The data qudits in time step $t=1$ are also adjacent to \emph{other} ancillas, not pictured, which are used to measure the $X$ stabilizers. Thus the stabilizers of the form \eqref{eq:qudit stab} centered at these data qudits should involve powers of $Z$ on those other ancillas! However, \emph{because the $X$ and $Z$ stabilizers must commute}, it is easy to check that these powers of $Z$ will cancel when we multiply the checks together as in the figure. 
(In the subsystem case, this same mechanism is why we only obtain detectors corresponding to \emph{stabilizers} rather than all gauge group generators.) 




\begin{figure}[H]
    \centering
    \includegraphics[width=.6\linewidth]{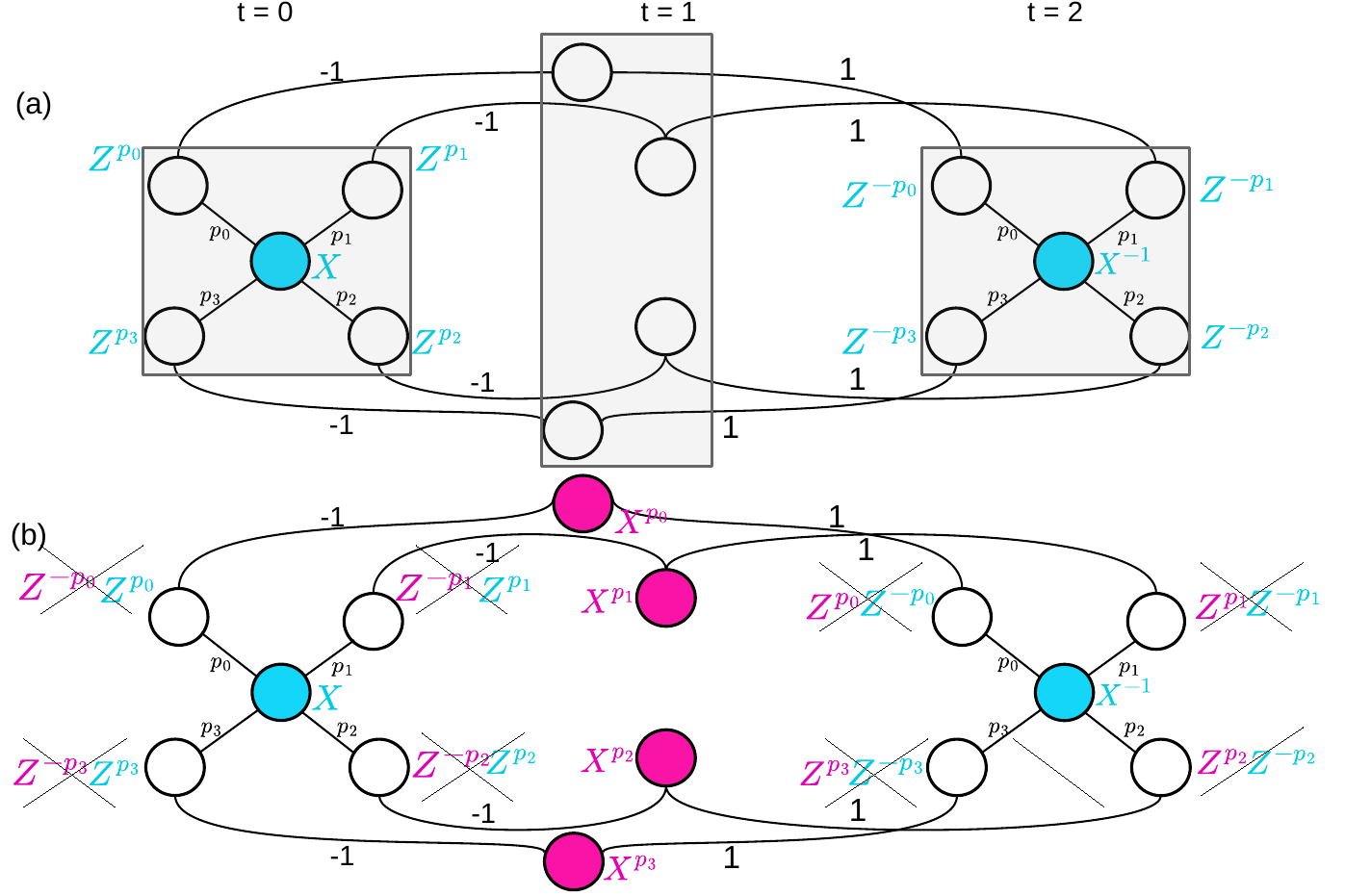}
    \caption{(a) A detector for foliated qudit MBQC. Here $Z_0^{p_0} \dots Z_{3}^{p_3}$ is the stabilizer of the code we are foliating. We multiply the corresponding graph state stabilizers in times $t=0$ and $t=2$ (blue), along with the stabilizers connecting these two time steps. (b) The same product of stabilizers with the cancellation performed. We now see that the remaining stabilizer contains only $X$'s and therefore commutes with the measurements we are making.}
    \label{fig:csschecks}
\end{figure}
\subsection{General example: Non-CSS measurements}\label{sec:non css example}
Next we discuss the measurement of non-CSS checks. As an example, we consider a check of the form $S=X^p Z^q\otimes Z^r$. 
We wish to start with a graph state, such as the one depicted in Fig.~\ref{fig:complexobservable}(a), perform single-qudit measurements, and obtain a state stabilized by $S$, as in Fig.~\ref{fig:complexobservable}(b). 
We show that this is the appropriate graph state and obtain the measurement rule of Definition~\ref{def:non css}. 

\begin{figure}[H]
    \centering
    \includegraphics[width=.6\linewidth]{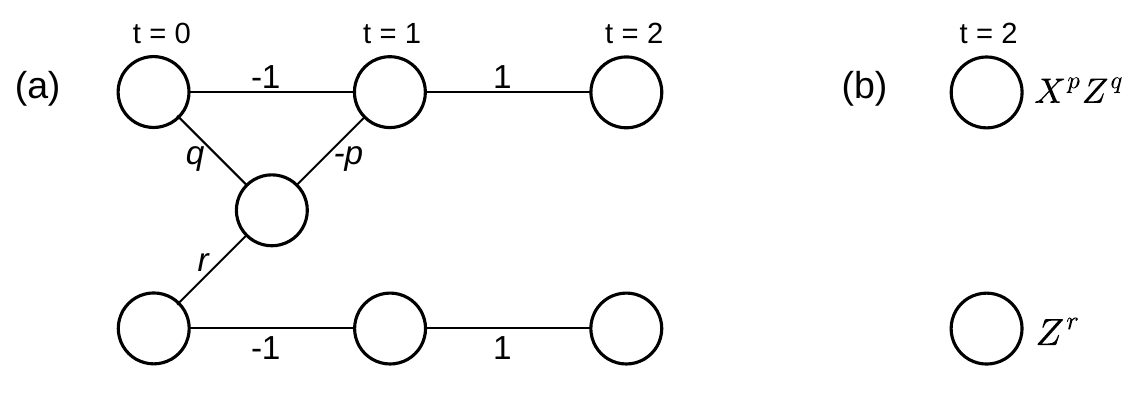}
    \caption{An example of a non-CSS check such as $X^p Z^q \otimes Z^r$. (a) is the initial graph state and (b) is the desired output.}
    \label{fig:complexobservable}
\end{figure}
As shown in Fig.~\ref{fig:solution}(a), we begin with the stabilizer centered at the central (ancilla) qudit, shown in blue.
Similarly to the previous section, we want to obtain a stabilizer of the graph state that has no $Z$'s on the data qudits in $t=0,1$. 
Thus we multiply by (the $r$th power of) the stabilizer centered at the red node in Fig.~\ref{fig:solution}(b). This solves the problem for the lower data qudit. 
We also attempt a similar cancellation for the upper data qudit, multiplying by (the $q$th power of) the stabilizer centered at the green node in Fig.~\ref{fig:solution}(c). 
This finishes cancelling the $Z$'s on the data qudits in $t=0$. 
Since we also need to cancel the $Z^{-p}$ on the upper data qudit in $t=1$, we also multiply by (the $p$th power of) the stabilizer centered at the purple node in Fig.~\ref{fig:solution}(d). 
As shown there, this gives a stabilizer of the graph state that commutes with the $X$ measurements we will make on the data qudits in $t=0,1$. 
However, unlike in the non-CSS case, the stabilizer shown in Fig.~\ref{fig:solution}(d) \emph{does not} have only a power of $X$ on the ancilla qudit! When we multiplied by the stabilizer in part (c), we ended up with a $Z^{-pq}$ factor on the ancilla. 
This is why Definition~\ref{def:non css} tells us to change the measurement to $XZ^{-pq}$. 
This change is nontrivial only when our ancilla is adjacent to the same data qudit chain state in both even and odd time steps, which can only happen in the non-CSS case. 

Once the data qudits in $t=0,1$ are measured in the $X$ basis, and the ancilla is measured in the $XZ^{-pq}$ basis, we obtain a state stabilized by $X^p Z^q \otimes Z^r$, as in Fig.~\ref{fig:complexobservable}(b). 

\begin{figure}[H]
    \centering
    \includegraphics[width=.6\linewidth]{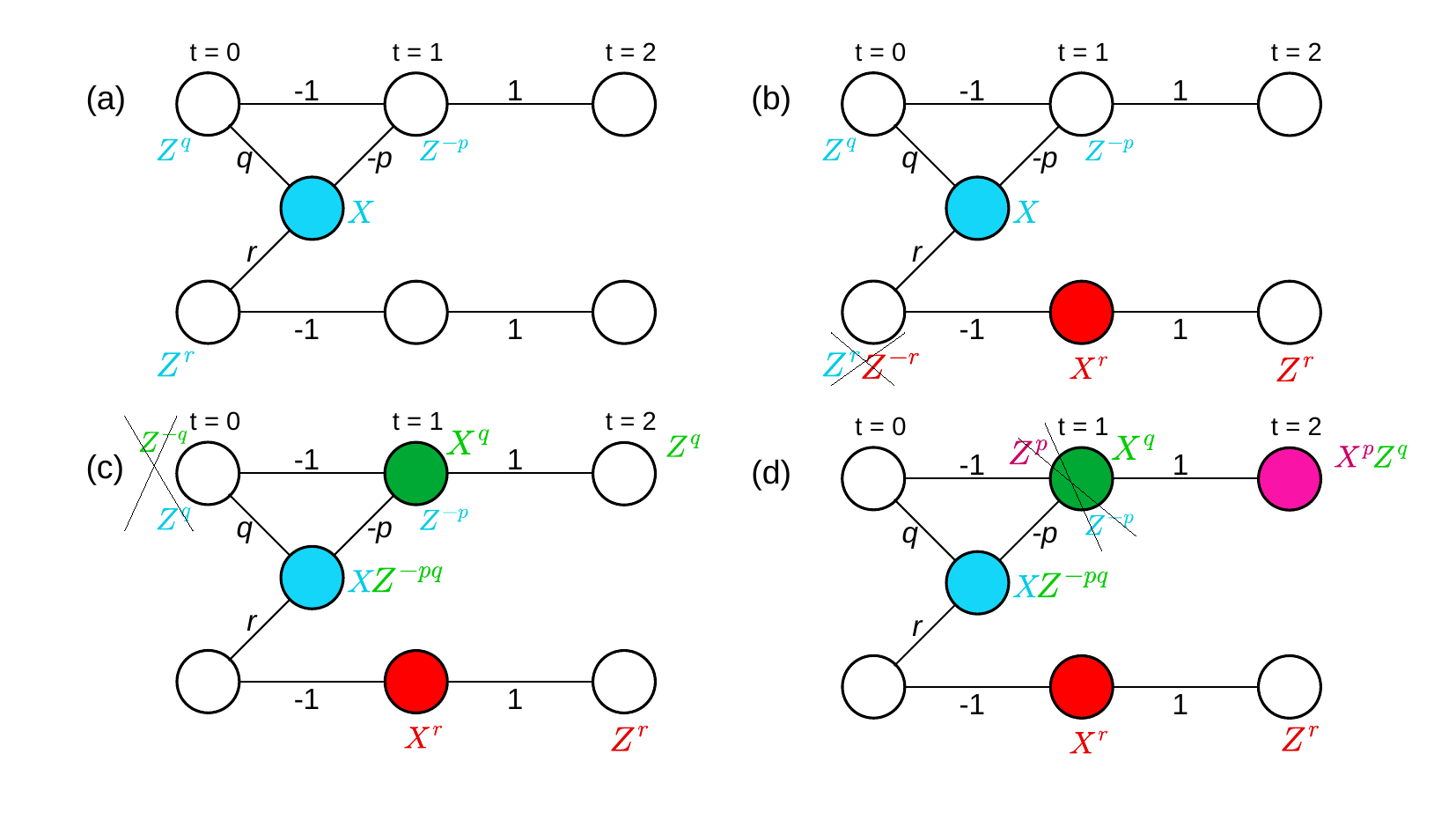}
    \caption{Step by step calculation for targeted observable  $X^p Z^q \otimes Z^r$}
    \label{fig:solution}
\end{figure}



\subsection{Example: Foliated Qudit toric code}
We briefly discuss the detectors for the foliated qudit toric code, a direct application of Sec.~\ref{sec:detector general example}. 
Fig.~\ref{fig:toricparity} shows the calculation of a detector for the $Z$-type stabilizers, of the form $Z\otimes Z\otimes Z^\dagger \otimes Z^\dagger$ (read counterclockwise around a face). 
Recall that the green edges indicate weight $-1$, with weight $+1$ otherwise. The presence of the weight $-1$ edges in $t=0$ allows for the daggers to appear in the right places in the stabilizer, as in part (a). The cancellation in part (b) is enabled by the fact that the edges from $t=0$ to $t=1$ have the opposite sign compared to the edges from $t=1$ to $t=2$. 


\begin{figure}
    \centering
    \includegraphics[width=.6\linewidth]{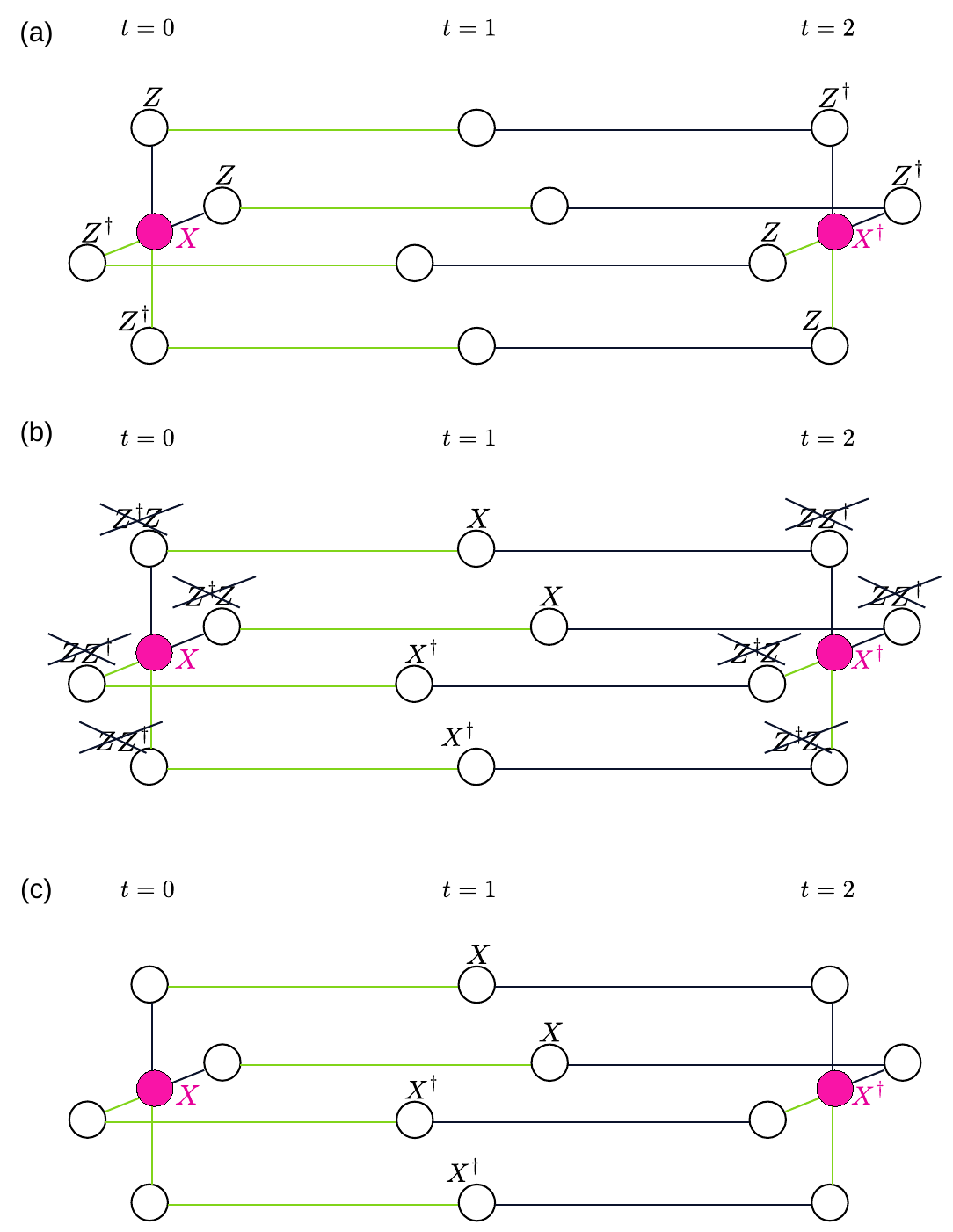}
    \caption{A $Z$-type detector in the toric code. Green edges have weight -1 and darker edges have weight $+1$. Pink circles indicate ancillas. Part (a) shows stabilizers of the graph state centered around the ancillas at $t=0$ and $t=2$, part (b) shows how we can cancel the Z factors, and part (c) shows the final detector.}
    \label{fig:toricparity}
\end{figure}

\begin{figure}
    \centering
    \includegraphics[width=.6\linewidth]{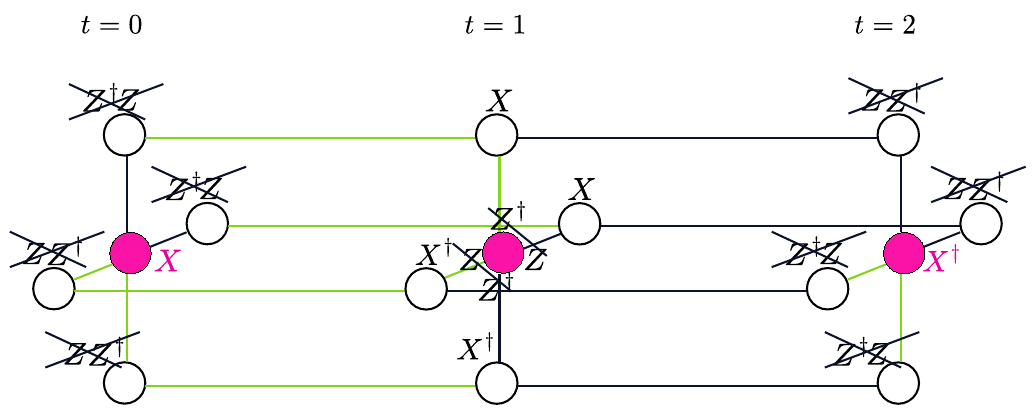}
    \caption{A $Z$-type detector for the $2\times 2$ toric code, with the $t=1$ ancilla included. }
    \label{fig:checkwithancilla}
\end{figure}

As an even more concrete example, we consider the $2\times 2$ toric code in Fig.~\ref{fig:checkwithancilla}. Here we are able to show the entire foliated graph state from $t=0$ to $t=2$. The ancilla at $t=1$ is not used for the $Z$-type detector; note how the $Z$'s on that ancilla cancel, because the $X$ and $Z$-type checks commute. 



For completeness, we also review the case of the $X$ checks (which have the form $X^\dagger \otimes X\otimes X\otimes X^\dagger$ when read counterclockwise) in Fig.~\ref{fig:Xcheck1}. We depict only the part of the graph state from $t=1$ to $t=3$. 
Except for the edge weights, we note this is identical to the $Z$-check case: the detector begins and ends with $X^{\pm 1}$ on the ancillas, with appropriate powers of $X$ on the data qudits in the middle. 
The symmetry between the two types of checks, and the translation-invariance of the toric code, gives rise to the standard cubic ``unit cell'' depiction of the foliated toric code, as in Fig.~\ref{fig:examples}(a). 




\begin{figure}
    \centering
    \includegraphics[width=.6\linewidth]{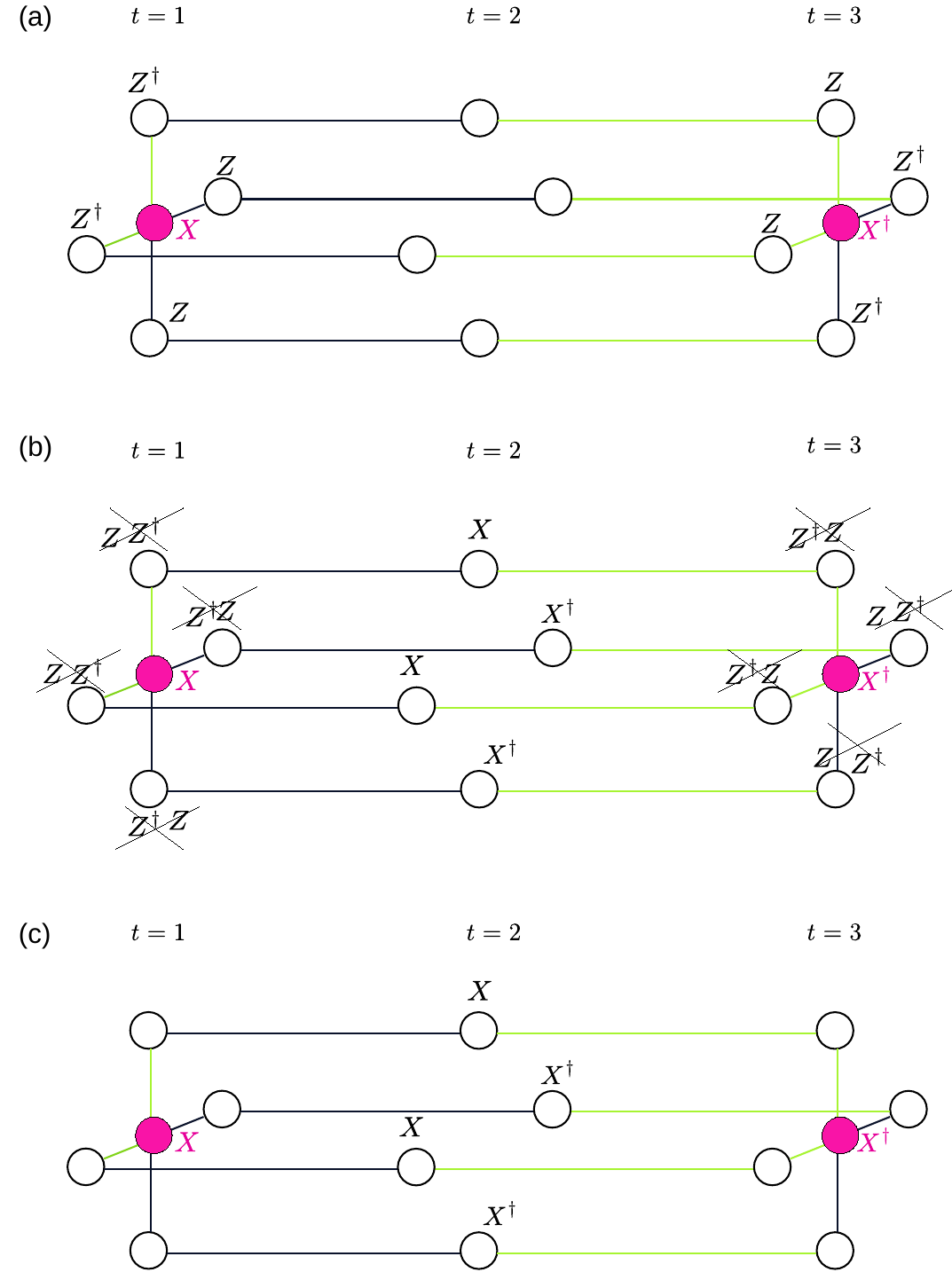}
    \caption{An $X$-type detector in the toric code. Green edges have weight -1 and darker edges have weight $+1$. Pink circles indicate ancillas. Part (a) shows stabilizers of the graph state centered around the ancillas at $t=0$ and $t=2$, part (b) shows how we can cancel the Z factors, and part (c) shows the final detector.
    }
    \label{fig:Xcheck1}
\end{figure}
\subsection{Example: Foliated qudit perfect code}\label{sec:perfect}

\begin{figure}[H]
    \centering
    \includegraphics[width=.6\linewidth]{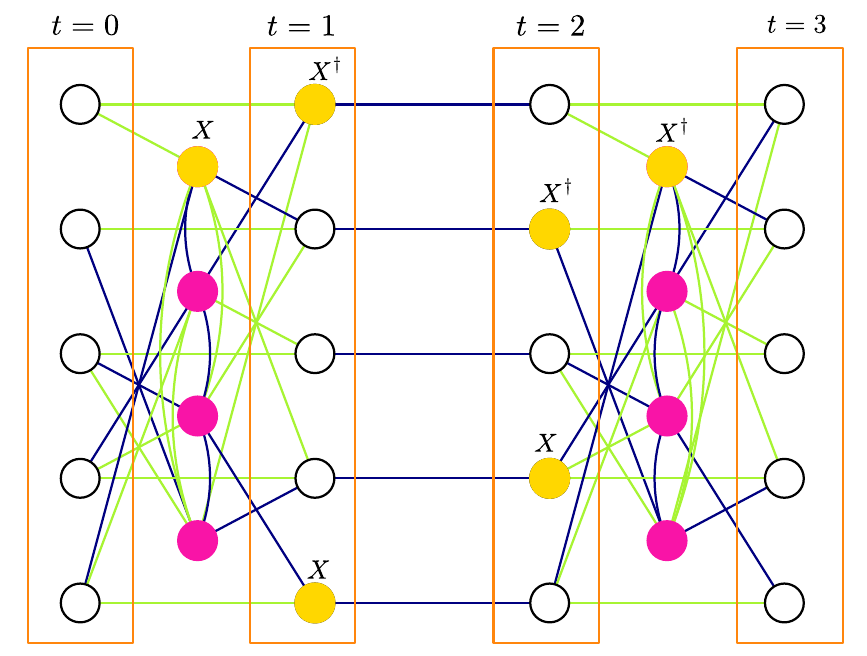}
    \caption{Foliated 5-qudit perfect code. The green edges have weight $-1$, while the dark blue edges have weight $+1$. 
    The boxed nodes are data qudits; the others are ancillas. 
    The yellow nodes are those contributing to the first detector in the code.}
    \label{fig:quditperfectcode}
\end{figure}
Here, we discuss the foliated qudit perfect code as an example of foliation for non-CSS stabilizer codes. In this case, the stabilizers are not separated into X and Z types. 
In particular, we measure the following stabilizers: 
\begin{equation}\label{eq:perfect code stabs}
    Z_0^\dagger X_1^\dagger I_2 X_3 Z_4, 
    X_0^\dagger I_1 X_2 Z_3 Z_4^\dagger, 
    I_0 X_1 Z_2 Z_3^\dagger X_4^\dagger, 
    X_0 Z_1 Z_2^\dagger X_3^\dagger I_4.
\end{equation}
In the formalism of Definition~\ref{def:non css}, we have, for all measurement rounds $k$, 
\begin{equation}
    \tilde{H}(k) = 
    \begin{bmatrix}
        -1 & 0 & 0 & 0 & 1 & 0 & -1 & 0 & 1 & 0\\
        0 & 0 & 0 & 1 & -1 & -1 & 0 & 1 & 0 & 0\\
        0 & 0 & 1 & -1 & 0 & 0 & 1 & 0 & 0 & -1\\
        0 & 1 & -1 & 0 & 0 & 1 & 0 & 0 & -1 & 0
    \end{bmatrix}.
\end{equation}
Figure~\ref{fig:quditperfectcode} shows the first two measurement rounds of the foliation. The boxed nodes represent the data qudits, with $t=0,1$ corresponding to the first measurement round and $t=2,3$ corresponding to the second measurement round. The qudits located between time steps 0 and 1, and between time steps 2 and 3, are the ancilla qudits belonging to the first and second measurement rounds respectively. 
From top to bottom, the ancillas in each measurement round correspond to the stabilizers of \eqref{eq:perfect code stabs}. 
Unlike the CSS case, there are also edges between the ancilla qudits. For convenience, we provide the adjacency matrix of the subgraph induced by the ancillary nodes, with the columns following the same order as \eqref{eq:perfect code stabs}: 
\begin{equation*}
    \begin{bmatrix}
        0 & 1 & -1 & -1 \\
        1 & 0 & 1 & -1 \\
        -1 & 1 & 0 & 1 \\
        -1 & -1 & 1 & 0
    \end{bmatrix}.
\end{equation*}

In the figure, the yellow nodes are those involved in the detector corresponding to the stabilizer $Z_0^\dagger X_1^\dagger I_2 X_3 Z_4$, spanning the first and second measurement rounds. 
In this case, the detector involves only powers of $X$ (although this does not always hold for the ancillary qudits in the non-CSS case). 
We note that this detector has the standard form of \eqref{eq:non css detector}, straightforwardly generalizing \eqref{eq:detector explicit}. There is an $X$ on the ancilla in the first measurement round and an $X^\dagger$ on the ancilla in the second measurement round. The data qudits in the support are given as follows: if the initial ancilla is adjacent to the $j$th data qudit at time $t=t_0$, then the detector is supported on the $j$th data qudit at time $t=t_0+1$. 
Since the ancilla is adjacent to data qudits with $t=0,1$, this means the support includes data qubits with $t=1,2$. 

In the ancillary files of the arXiv version of this paper, we give the adjacency matrix for the foliated code with two measurement rounds, write down the four relevant detectors explicitly, present some randomly sampled data indicating that the value of each detector is always $1$ in the absence of errors. 
\subsection{Example: Foliated Qudit CSS Honeycomb code}
As an example of a dynamical (Floquet) code, we now consider the CSS honeycomb code \cite{davydova2023floquet} in the qudit case. 
As mentioned in App.~\ref{app:honeycomb}, the CSS honeycomb code may be extended to prime-dimensional qudits by simply replacing the $Z\otimes Z$ checks with $Z\otimes Z^\dagger$ checks. 
The data qudits of the code live on a $3$-colored toric hexagonal lattice as in Fig.~\ref{fig:hex lattice supplement}. 
Extending the notation of \cite{davydova2023floquet}, the period-6 measurement schedule can be written as $gZZ^\dagger, bXX, rZZ^\dagger, gXX, bZZ^\dagger, rXX$, where $gZZ^\dagger$ means to measure $Z\otimes Z^\dagger$ on the green edges, etc. 

For the foliated CSS honeycomb code, each detector spans five time steps (as opposed to three in the stabilizer/subsystem case) and involves two colors. 
In Fig.~\ref{fig:redface} and the rest of this section, we will consider in detail the detector associated to a red face of the lattice, starting at $t=0$. 
Specifically, we will consider a $Z$-stabilizer associated with a red face of the lattice, which has the form $S=Z\otimes Z^\dagger \otimes Z\otimes Z^\dagger \otimes Z\otimes Z^\dagger$ when read counterclockwise. 
In terms of the non-foliated code, this stabilizer is first measured during the $gZZ^\dagger$ step, as it can be obtained by multiplying the values of the $Z\otimes Z^\dagger$ checks on the three green edges of a red face. 
The stabilizer $S$ then commutes with the next three rounds of measurements, even though the individual $gZZ^\dagger$ gauge measurements do not commute with the $bXX$ and $gXX$ measurements. 
(Note that the checks measured in the $rZZ^\dagger$ cannot be used to obtain the value of $S$.) 
Finally, the $bZZ^\dagger$ step measures $S$ again. 
This will correspond directly to the shape of our detector, shown in Fig.~\ref{fig:redface}. 
The detector begins at $t=0$ when the $gZZ^\dagger$ checks are measured and is supported on the three (green) ancillas in that time step. It then commutes through the three intervening rounds and ends at $t=4$ when the $bZZ^\dagger$ checks are measured (supported on the three blue ancillas at $t=4$). 
The powers of $X$ on the data qudits in $t=1$ and $t=3$ serve to appropriately cancel the $Z$'s appearing on the data qudits in $t=0,4$. 
We note that Fig.~\ref{fig:redface} shows no nodes in the $t=2$ time step, since the detector is not supported there: the $Z$'s produced by the $X$'s in $t=1,3$ all cancel. 

\begin{figure}
    \centering
    \includegraphics[width=.8\linewidth]{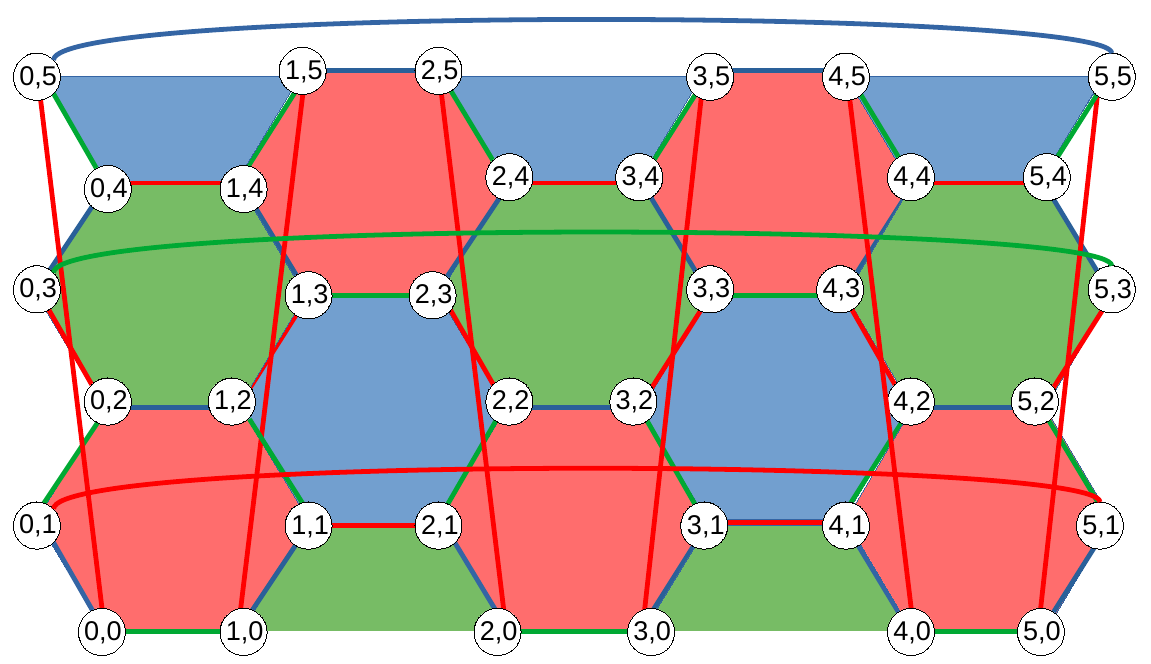}
    \caption{The toric hexagonal lattice used for the CSS honeycomb code \cite{davydova2023floquet} and our qudit generalization.}
    \label{fig:hex lattice supplement}
\end{figure}

\begin{figure}[H]
    \centering
    \includegraphics[width=\linewidth]{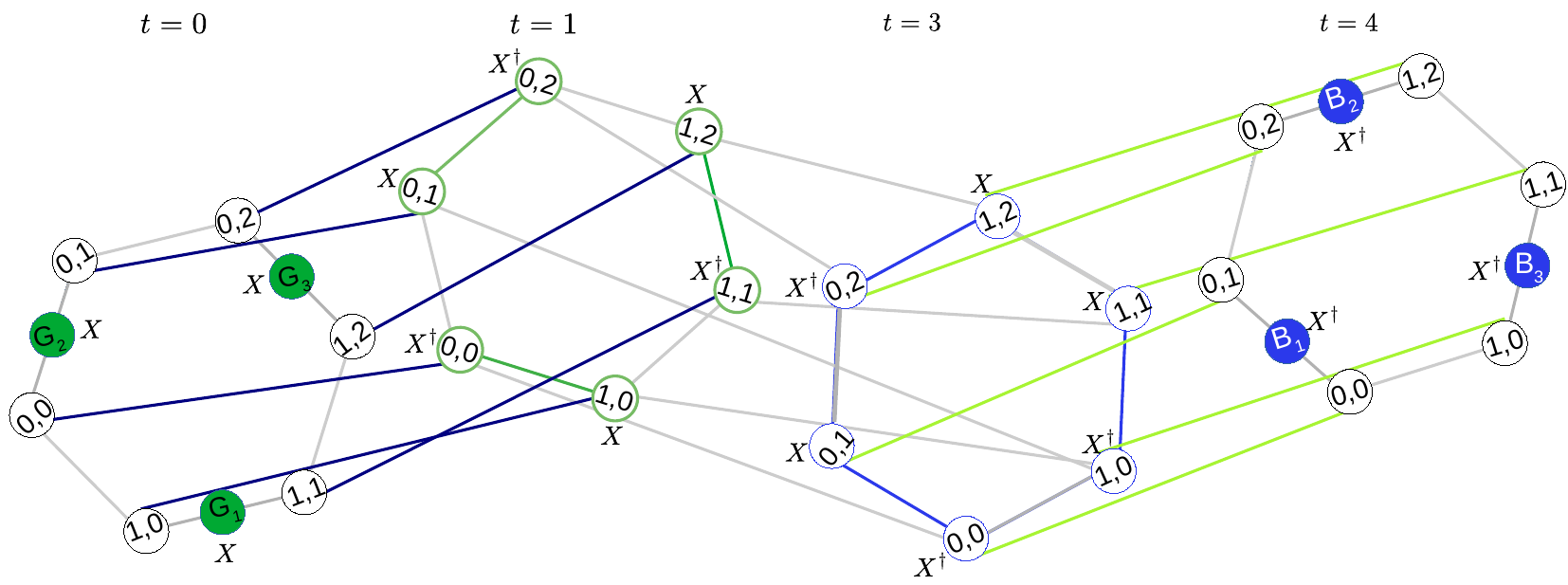}
    \caption{The detector starting at $t=0$ associated with a red face in the foliated CSS honeycomb code (for qudits). The detector begins with $3$ ancilla qudits in $t=0$ (corresponding to the $3$ $gZZ^\dagger$ checks whose values must be multiplied to obtain the stabilizer $S$), and similarly ends with $3$ ancilla qudits in $t=4$ (corresponding to the $3$ $bZZ^\dagger$ checks whose values must be multiplied to obtain $S$). 
    Only those nodes and edges necesssary to understand the detector structure are shown: in particular, the (irrelevant) ancilla qudits in $t=1,3$ and all qudits in $t=2$ are omitted. 
    }
    \label{fig:redface}
\end{figure}

\end{document}